\documentclass[11pt,a4paper,reqno]{article}

\usepackage{amssymb}
\usepackage{latexsym}
\usepackage{amsmath}
\usepackage{graphicx}
\usepackage{amsthm}
\usepackage{empheq}
\usepackage{bm}
\usepackage{booktabs}
\usepackage[dvipsnames]{xcolor}
\usepackage{pagecolor}
\usepackage{caption}
\usepackage{subcaption}
\usepackage{tikz,pgfplots}
\usetikzlibrary{matrix,decorations.pathreplacing,positioning}
\usepackage{enumitem}
\usepackage{wasysym}
\usepackage{mdframed}
\usepackage{ragged2e}
\usepackage{multirow}
\usepackage{array}
\usepackage{hhline}
\usepackage{enumitem}
\usepackage{hyperref} 
\usepackage{arydshln}
\usepackage{centernot}
\usepackage{mathtools}
\usepackage{stmaryrd}
\usepackage{relsize}
\usepackage{tikz,pgfplots}
\usetikzlibrary{arrows.meta}
\usetikzlibrary{positioning}
\usetikzlibrary{chains}
\usetikzlibrary{decorations.pathmorphing}

\usepackage{soul}
\usepackage{calrsfs}
\usepackage{float}
\usepackage{wrapfig}

\newcommand{\rmv}[1]{}
\usepackage[margin=3cm]{geometry}

\numberwithin{equation}{section}

\newtheorem{theorem}{Theorem}[section]
\newtheorem{corollary}[theorem]{Corollary}
\newtheorem{proposition}[theorem]{Proposition}
\newtheorem{lemma}[theorem]{Lemma}

\theoremstyle{definition}
\newtheorem{definition}[theorem]{Definition}

\newtheorem{remark}[theorem]{Remark}

\newcommand{\F}{\mathbb{F}}
\newcommand{\R}{\mathbb{R}}
\newcommand{\mN}{\mathcal{N}}
\newcommand{\alp}{\mA}
\newcommand{\ad}{\bold{A}}
\newcommand{\U}{\mathcal{U}}
\newcommand{\mF}{\mathcal{F}}
\newcommand{\mV}{\mathcal{V}}
\newcommand{\mE}{\mathcal{E}}
\newcommand{\N}{\mathbb{N}}
\newcommand{\C}{\textup{C}}
\newcommand{\Fq}{\F_q}
\newcommand{\bS}{S}
\newcommand{\bT}{\bold{T}}

\newcommand{\mU}{\mathcal{U}}
\newcommand{\mA}{\mathcal{A}}
\newcommand{\mD}{\mathcal{D}}
\newcommand{\mS}{\mathcal{S}}
\newcommand{\h}{\textup{H}}
\newcommand{\concat}{\RHD} 

\newlength{\mynodespace}
\definecolor{myg}{RGB}{220,220,220}

\usepackage{xcolor}
\usepackage{soul}

\usepackage{hyperref}
\usepackage{authblk}

\title{Multishot Capacity of Networks with Restricted Adversaries}

\author[1]{Giuseppe Cotardo\thanks{G. C. was partially supported by the NSF grants DMS-2037833 and DMS-2201075, and by the Commonwealth Cyber Initiative.}}\affil[1]{Virginia Tech, Blacksburg, U.S.A.}

\author[1]{Gretchen L. Matthews\thanks{G. L. M. is supported by NSF DMS-2201075 and the Commonwealth Cyber Initiative. }}

\author[2]{Alberto Ravagnani\thanks{A. R. is supported by the Dutch Research Council through grants OCENW.KLEIN.539 and VI.Vidi.203.045.}}
\affil[2]{Eindhoven University of Technology, the Netherlands}

\author[1]{Julia Shapiro\thanks{J. S. is supported by the Department of Defense Cyber Service Academy Scholarship.}}

\date{}

\begin{document}

\maketitle

\begin{abstract}
We investigate adversarial network coding and decoding, focusing on the multishot regime and when the adversary is restricted to operate on a vulnerable region of the network. Errors can occur on a proper subset of the network edges and are modeled via an adversarial channel. The paper contains both bounds and capacity-achieving schemes for the Diamond Network, the Mirrored Diamond Network, and generalizations of these networks. We also initiate the study of the capacity of 3-level networks in the multishot setting by computing the multishot capacity of the Butterfly Network, considered in [IEEE Transactions on Information Theory, vol. 69, no. 6,~2023], which is a variant of the network introduced by Ahlswede, Cai, Li and Yeung in 2000.
\end{abstract}

\medskip

\section{Introduction}
Network coding is a communication strategy where the intermediate nodes of a network are allowed to perform coding operations on the data packets they receive before sending them towards the sinks. This strategy, in contrast to simple forwarding, can improve network throughput and optimize resource utilization. Network coding was proposed in 2000 by Ahlswede, Cai, Li and Yeung in~\cite{network2000flow} and is a thriving area of research; see~\cite{yueng2006upper,random2008network,medard2003,kschischang2019multi,jafari2009multi, multi2009bound, nutman2008, kotter2008,   wang2007broadcast} among many others. Adversarial models have been investigated using network coding in many contexts, including Byzantine adversaries~\cite{byzantine2007} and adversaries able to control some of the networks vertices~\cite{wang2007broadcast}. Error correction in the context of linear network coding and bounds for network coding have been studied extensively; see for instance~\cite{yueng2006upper, yueng2006lower, linear2003network, yueng2007l, yueng2007refined}.

Following~\cite{ravagnani2018}, we consider networks affected by adversarial noise. When the adversary is not restricted, end-to-end coding combined with network coding can be used to achieve capacity. End-to-end approaches for controlling random and coherent errors were presented in~\cite{random2008network, kotter2008}, including efficient encoding and decoding schemes using rank-metric codes. In \cite{multiple,kschischang2019multi, multi2009bound,  nobrega2010multi}, the authors have studied multishot network coding in the context where end-to-end coding combined with network coding can be used to achieve capacity. 

In~\cite{beemer2023network, BEEMER202236, curious2021diamond}, it was shown that restricting the adversary to a proper subset of the network edges makes end-to-end communication strategies suboptimal. In this context, the concept of \textit{network decoding}, the use partial decoding in the intermediate nodes, was introduced as an essential strategy for achieving capacity in networks with restricted adversaries. The authors demonstrated a minimal example of a network, called the diamond network, with a restricted adversary, where both  the Singleton Cut-Set bound introduced in \cite{ravagnani2018} (the best cut-set bound for this network) cannot be achieved, and the value of the one-shot capacity cannot be achieved using linear operations at intermediate nodes to process information before forwarding. This network requires network decoding to achieve capacity. The Mirror Diamond Network was then introduced as an extension of the Diamond Network and the capacity of this network was shown to meet the Singleton Cut-Set bound, but still forced to use non-linear network codes to achieve capacity. By limiting where the adversary can attack to a proper subset of network edges, the best known cut-set bounds for networks are not always tight. This restriction regardless of the alphabet size also destroys the ability to achieve capacity with linear network codes in combination with end-to-end coding/decoding.

In this paper, we advance the study of networks with restricted adversaries and their capacity. We focus on restricted adversaries, aligning with the framework proposed in~\cite{beemer2023network}, which requires network decoding to achieve one-shot capacity. The goal of this paper is to assess the potential increase in capacity achieved by utilizing a network \textit{multiple} times, instead of just once as in~\cite{beemer2023network}. While the benefits of multiple channel uses are clear and well understood in the context of classical coding theory, for networks with restricted adversaries, these benefits have never been investigated before. The only contribution we are aware of is our preliminary work~\cite{multishotND}, which serves as the starting point for this paper.

The rest of the paper is organized as follows: Sections~\ref{prob} and ~\ref{prelim} introduce the problem set up and necessary notation. Section~\ref{multishot1} computes the capacity of the Diamond Networks and the families of networks introduced in~\cite{beemer2023network, BEEMER202236,curious2021diamond}. Section~\ref{doubleM} establishes an upper bound on the capacity of the larger network based on that of the corresponding smaller network. The bounds in Section~\ref{doubleM} are then used to compute the capacity of the butterfly network in Figure \ref{fig:butt} over multiple uses. The paper ends with a conclusion given in Section~\ref{S:conclusion} and an appendix containing some additional details.

\section{Motivating Example}\label{prob}

Our work focuses on the typical setting in network coding, as in~\cite{network2000flow,ravagnani2018,linear2003network} and many others. Our goal is to understand what the advantage is of using a network multiple times for communication in a setting where the adversary is restricted to corrupting a proper subset of network edges. This is a nontrivial question that has not been addressed and was first introduced in \cite{beemer2023network}. With some strategies, there is an immediate gain of capacity provided over the multiple transmission rounds. We start with an example from~\cite{beemer2023network}, where we illustrate the advantage of using a network multiple times by providing the capacity of the network $\mathcal{B}$ in Figure \ref{fig:butt}.

\begin{figure}[h!]
\centering
	\begin{tikzpicture}[
      mycircle/.style={
         circle,
         draw=black,
         fill=white,
         fill opacity = 0.4,
         text opacity=1,
         inner sep=0pt,
         minimum size=20pt,
         font=\small},
      ->,>=stealth,thick,
      node distance=1.2cm and 1.6cm,
      every text node part/.style={align=center}
      ]
		\node[mycircle,fill=white] (1) {$S$};
		\node[mycircle, below right =of 1,fill=Salmon!60] (2) {$V_2$};
		\node[mycircle, above right =of 1,fill=Salmon!60] (3) {$V_1$};
		\node[mycircle, below right =of 3,fill=Salmon!60] (4) {$V_3$};
		\node[mycircle, right =of 4,fill=Salmon!60] (5) {$V_4$};
		\node[mycircle, below right =of 5,fill=white] (6) {$T_2$};
		\node[mycircle, above right =of 5,fill=white] (7) {$T_1$};
		
		\path[every node/.style={font=\sffamily\small}]
    		(1) edge[bend right,dashed] node [fill=white,sloped,inner sep=2pt] {$e_4$} (2)
    		(1) edge[bend left,dashed] node [fill=white,sloped,inner sep=2pt] {$e_3$} (2)
    		(1) edge[bend right,dashed] node [fill=white,sloped,inner sep=2pt] {$e_2$} (3)
    		(1) edge[bend left,dashed] node [fill=white,sloped,inner sep=2pt] {$e_1$} (3)
    		(2) edge[dashed] node[fill=white,sloped,inner sep=2pt] {$e_7$} (4)
    		(2) edge node[fill=white,sloped,inner sep=2pt] {$e_8$} (6)
    		(3) edge[dashed] node[fill=white,sloped,inner sep=2pt] {$e_6$} (4)
    		(3) edge node[fill=white,sloped,inner sep=2pt] {$e_5$} (7)
    		(4) edge[dashed] node[fill=white,sloped,inner sep=2pt] {$e_9$} (5)
    		(5) edge node [fill=white,sloped,inner sep=2pt] {$e_{11}$} (6)
    		(5) edge node [fill=white,sloped,inner sep=2pt] {$e_{10}$} (7)
    		;
	\end{tikzpicture}
 \caption{The Butterfly Network $\mathcal{B}$.}
 \label{fig:butt}
\end{figure}
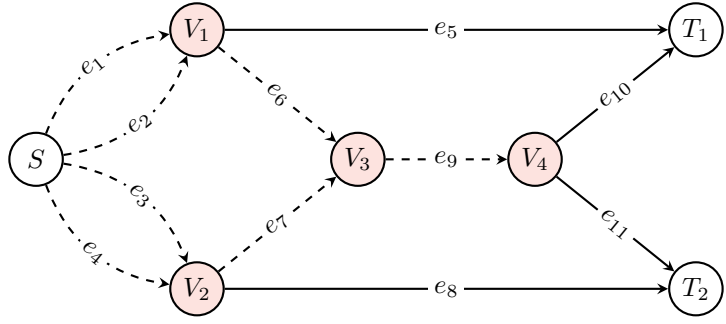

Classically, using a channel multiple times may increase the capacity; see~\cite[Example~56]{ravagnani2018}. In~\cite[Theorem~8.9]{beemer2023network}, the authors proved that the one-shot capacity of $\mathcal{B}$ is~$\log_{|\alp|}(|\alp|-1)$ with the help of the Double Cut-Set Bound in~\cite[Theorem 8.6]{beemer2023network}. We extend~\cite[Theorem 8.6]{beemer2023network} to the multishot setting in Theorem \ref{dcsb}. In Theorem~\ref{lbn} we compute the capacity of $\mathcal{B}$ over multiple times in the setting where the adversary is more restricted, meaning, it cannot change the edges attacked each transmission round. The \textit{$i$-shot capacity of $\mathcal{B}$} shown in Theorem~\ref{lbn} is \begin{equation*}
\frac{\log_{|\alp|}(|\alp|^i-1)}{i}.
\end{equation*}

In this restricted adversarial model, there is a gain in using $\mathcal{B}$ multiple times for communication. In contrast, when the adversary is free to change the edges attacked each transmission round, we show in Theorem~\ref{cbn}, that the \textit{$i$-shot capacity of $\mathcal{B}$} is
\begin{equation*}
\log_{|\alp|}(|\alp|-1),
\end{equation*} the same as the one-shot capacity. When the adversary is more free, there is no gain in using $\mathcal{B}$ multiple times for communication. This contrast and determining whether there is a gain in capacity over multiple uses of a network motivates our work.

\section{Network Decoding}\label{prelim}
Throughout the paper, we let $q$ be a prime power, $\Fq$ be the finite field with $q$ elements, and $d_{H}(x,y)$ denote the Hamming distance between $x$ and $y$, where $x$ and $y$ are codewords in a code $C$. In this section, we recall some definitions and notation from network decoding for the convenience of the reader and also point the reader to Appendix~\ref{sec:app} and additional references such as~\cite{beemer2023network,ravagnani2018}.

\begin{definition} A (single source) \textbf{network} is a 4-tuple $\mN = (\mV,\mE,S,\mathbf{T})$ where:
\begin{enumerate}
\item $(\mV,\mE)$ is a directed, acyclic and finite multigraph;\vspace{-2ex}
\item $S \in \mV$ is the  \textbf{source};\vspace{-2ex}
\item $\bold{T} \subseteq \mV$ is the set of \textbf{terminals}.
\end{enumerate}
We also assume the following:
\begin{enumerate}[resume]
\item $|\bold{T}| \geq 1$ and $S \notin \bold{T}$.\vspace{-2ex}
\item The source does not have incoming edges and the terminals do not have outgoing edges.\vspace{-2ex}
\item There exists a directed path from $S$ to any $T \in \bold{T}$.\vspace{-2ex}
\item For every $V \in \mathcal{V}\setminus (\{S\} \cup \bold{T}),$ there exists a directed path from $S$ to $V$ and from $V$ to $T$ for some $T \in \bold{T}$.
\end{enumerate}
\end{definition}

The elements $V \in \mathcal{V}\setminus (\{S\}\cup \bold{T})$ are called \textbf{intermediate nodes}. We denote the sets of incident edges as $\textup{in}(V)$ and $\textup{out}(V)$ with their cardinalities defined as the \textbf{indegree}~$\deg^{-}(V)$ and \textbf{outdegree} $\deg^{+}(V)$. Each edge of the network~$\mN$ carries at most one element from an \textbf{alphabet} $\mA$, with $|\alp| \geq 2$. The intermediate vertices $V$ in the network~$\mN$ receive symbols from $\mA$ over the incoming edges, process them according to a chosen set of functions, and then output the information over the outgoing edges. We model errors in transmission as presented by an omniscient adversary $\ad_{\mN}$ who can change the symbol sent on up to $t$ edges from a fixed subset $\mathcal{U} \subseteq \mathcal{E}$ to any other symbol of $\alp$. The pair $(\mN, \ad_{\mN})$, where $\ad_{\mN}$ is the adversary, is called an \textbf{adversarial network}.

The edges $\mE$ of a network $\mN = (\mV, \mE, S, \bold{T})$ can be partially ordered as follows. Given~$e, e' \in \mE$, we say that $e$ precedes $e'$, denoted $e \preceq e'$, if there exists a directed path in~$\mN$ from $e$ to $e'$. Notice that the partial order $\preceq$ on $\mE$ can be extended to a total order~$\leq$. The total order extension satisfies the property $e \preceq e'$ implies $e \leq e'$. We assume a fixed total order, indicated by the labeling of the edges.

\begin{definition} [{\cite[Definition 40]{ravagnani2018}}] Let $\mN=(\mV, \mE, S, \bold{T})$ be a network. A \textbf{network code} $\mathcal{F}$ for $\mN$ is a set of functions 
$\{\mathcal{F}_V\,:\, \mA^{\textup{deg}^{-}(V)} \to \mA^{\textup{deg}^{+}(V)} \mid  V \in 
  \mathcal{V} \setminus (\{S\} \cup \bold{T})\}.$ 
\end{definition} 

The functions in $\mF$ describe how $\mN$ processes the information at each intermediate node input from the incoming edges. These functions are unique due to the choice of the total order $\leq$.

\begin{definition}\label{notC}
Let $(\mN,\ad)$ be an adversarial network with $\mN=(\mV, \mE, S, \bold{T})$ and let~$\U \subseteq \mE$ be a set of edges that the adversary can corrupt. Let $\mF$ be a network code for~$\mN$ and $t \geq 0$ be an integer. The channel representing the transfer from $S$ to $T \in \bold{T}$ is
\[ \Omega[\mN, \mA, \mF, S \longrightarrow T, \mU, t] \,:\, \alp^{\textup{deg}^+(S)} \to \alp^{\textup{deg}^{-}(T)}.\]
\end{definition}

We often write $\Omega$ to represent $\Omega[\mN, \mA, \mF, S \longrightarrow T, \mU, t]$. We now recall the definition of an unambiguous code.

\begin{definition}
An \textbf{outer} \textbf{code} for a network $\mN=(\mV, \mE, S, \bold{T})$ is a subset
$C \subseteq \alp^{\textup{deg}^{+}(S)}$ with $|C| \ge 1$. An \textbf{unambiguous} code $C$ for each channel $\Omega$ is a code such that for all~$x,y \in C$ with $x \neq y$ and for all $T \in \bold{T}$ we have
\begin{equation*}
    \Omega[\mN, \alp, \mF, \textup{out}(S) \longrightarrow \textup{in}(T), \mU,t](x) \cap \Omega[\mN, \alp, \mF, S \longrightarrow T, \mU, t](y) =\emptyset.
\end{equation*}
\end{definition}

The empty intersection guarantees that every element of $C$ can be recovered by every terminal uniquely, regardless of the action the adversary takes. We note that for any unambiguous code, any subcode is also unambiguous. Next, we define the notion of one-shot capacity of an adversarial network as in~\cite[Definition~3.18]{beemer2023network}. Let $\ad_{\mN}$ be the adversary for $\mN$ capable of corrupting up to $t$ edges of a proper subset $\mU \subsetneq \mE.$ 

\begin{definition}\label{oneshot} The \textbf{one-shot} capacity of an adversarial network $(\mathcal{N},\ad_{\mathcal{N}})$ over an alphabet $\alp$, denoted $C_1(\mN, \alp, \ad_{\mN})$, is the maximum $\alpha \in \R$ such that there exists an unambiguous code $C$ for $\Omega$ and a network code $\mathcal{F}$ with $\alpha = \log_{|\mathcal{A}|}(|C|),$ meaning, \[C_1(\mN, \alp, \ad_{\mN}) = \max \left\{\log_{|\alp|}{|C|}\,:\, C \hbox{ is an unambiguous code for each } \Omega \right\}.\]

\end{definition}

An upper bound on the one-shot capacity of an adversarial network is given below. Recall that an edge cut of a graph $G$ is a set of edges $\mE' \subseteq \mE$ such that $G \setminus \mE'$ is disconnected.  

\begin{theorem}[{Singleton Cut-Set Bound~\cite[Corollary 66]{ravagnani2018}}]\label{cutset}
Let $t \geq 0$ and let $\alp$ be an alphabet. Suppose that an adversary $\ad_{\mN}$ can corrupt up to $t$ edges from a subset $\mU \subseteq \mE$. We have that
\[\displaystyle C_1(\mN, \alp, \ad_{\mN}) \leq \min_{T \in \bT} \min_{\mE'}\left(|\mE'\setminus \mU| + \max\{0,|\mE' \cap \mU|- 2t\} \right),\] where $\mE' \subseteq \mE$ ranges over all edge cuts between the source $S$
and $T$.   
\end{theorem}

We now discuss the previous work on the one-shot capacity of the Diamond Networks provided in~\cite{beemer2023network,BEEMER202236,curious2021diamond}. Let $\mD$ be the network in Figure~\ref{diamond} and let $\ad_{\mD}$ be an adversary that can corrupt at most one of the dashed edges. The pair $(\mD,\ad_{\mD})$ is called the \textbf{Diamond Network}.
\begin{figure}
    \centering
    \begin{tikzpicture}
\tikzset{vertex/.style = {shape=circle,draw,inner sep=0pt,minimum size=2em}}
\tikzset{nnode/.style = {shape=circle,fill=Salmon!30,draw,inner sep=0pt,minimum
size=2em}}
\tikzset{edge/.style = {->,> = stealth}}
\tikzset{ddedge/.style = {dashed,->,> = stealth}}
\node[vertex] (S) {$S$};
\node[shape=coordinate,right=\mynodespace of S] (L) {};
\node[nnode,above=0.5\mynodespace of L] (V1) {$V_1$};
\node[nnode,below=0.5\mynodespace of L] (V2) {$V_2$};
\node[vertex,right=\mynodespace of L] (T) {$T$};
\draw[ddedge,bend left=0] (S)  to node[sloped,fill=white, inner sep=1pt]{\small $e_1$} (V1);
\draw[ddedge,bend left=16] (S) to  node[sloped,fill=white, inner sep=1pt]{\small $e_2$} (V2);
\draw[ddedge,bend right=16] (S)  to node[sloped,fill=white, inner sep=1pt]{\small $e_3$} (V2);
\draw[edge,bend left=0] (V1)  to node[sloped,fill=white, inner sep=1pt]{\small $e_4$} (T);
\draw[edge,bend left=0] (V2)  to node[sloped,fill=white, inner sep=1pt]{\small $e_{5}$} (T);
\end{tikzpicture} 
\caption{\label{diamond}{{The Diamond Network $\mD$.}}}
\end{figure}
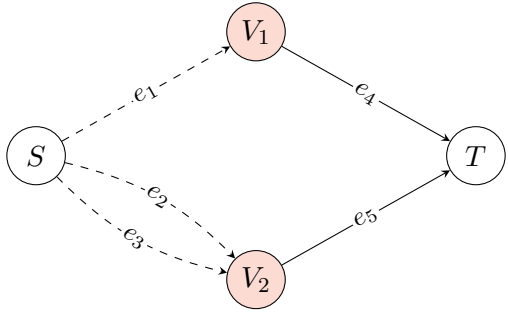
 It was shown in~\cite[Section~III]{BEEMER202236} that this is the smallest example of a network that does not meet the Singleton Cut-Set Bound~\cite[Corollary 66]{ravagnani2018}, illustrating the importance of intermediate nodes performing \textit{partial decoding} in order to achieve capacity. The authors show in~\cite[Theorem 13]{BEEMER202236} that the one-shot capacity of $\mD$ is $$C_1(\mD, \alp, \ad_{\mD}) = \textup{log}_{|\alp|}(|\alp|-1).$$ 
The strategy provided in~\cite{BEEMER202236} demonstrates that one symbol must be reserved from $\alp$ to implement an adversary detection mechanism, resulting in a non-integer value capacity. In~\cite[Sections~3 and~4]{BEEMER202236}, it was shown that the network in Figure~\ref{mirrored}, attains the Singleton Cut-Set Bound. The pair $(\mS,\ad_{\mS})$ is called the \textbf{Mirrored Diamond Network}. 
\begin{figure}
    \centering
    \begin{tikzpicture}
\tikzset{vertex/.style = {shape=circle,draw,inner sep=0pt,minimum size=2.0em}}
\tikzset{nnode/.style = {shape=circle,fill=Salmon!30,draw,inner sep=0pt,minimum
size=2.0em}}
\tikzset{edge/.style = {->,> = stealth}}
\tikzset{ddedge/.style = {dashed,->,> = stealth}}

\node[vertex] (S) {$S$};

\node[shape=coordinate,right=\mynodespace of S] (L) {};

\node[nnode,above=0.5\mynodespace of L] (V1) {$V_1$};

\node[nnode,below=0.5\mynodespace of L] (V2) {$V_2$};

\node[vertex,right=\mynodespace of L] (T) {$T$};

\draw[ddedge,bend left=16] (S)  to node[sloped,fill=white, inner sep=1pt]{\small $e_1$} (V1);

\draw[ddedge,bend right=16] (S) to  node[sloped,fill=white, inner sep=1pt]{\small $e_2$} (V1);

\draw[ddedge,bend left=16] (S)  to node[sloped,fill=white, inner sep=1pt]{\small $e_3$} (V2);

\draw[ddedge,bend right=16] (S)  to node[sloped,fill=white, inner sep=1pt]{\small $e_4$} (V2);

\draw[edge,bend left=0] (V1)  to node[sloped,fill=white, inner sep=1pt]{\small $e_4$} (T);

\draw[edge,bend left=0] (V2)  to node[sloped,fill=white, inner sep=1pt]{\small $e_{5}$} (T);

\end{tikzpicture} 
\caption{{{\label{mirrored} The Mirrored Diamond Network $\mS$.}}}
\end{figure}
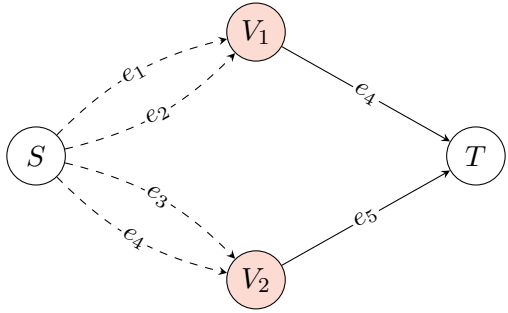
The authors show in~\cite[Proposition~14]{BEEMER202236} that the the one-shot capacity of $\mS$ is $$C_1(\mS, \alp, \ad_{\mS}) =1.$$ 

We follow the notation introduced in~\cite[Section~V]{beemer2023network} and refer the reader to that paper for details. A network $\mN = (\mV, \mE, \bS, \bT)$ is simple if $\bT = \{T\}$. Such a network $\mN$ is a simple $2$-level network if every path from $S$ to $T$ has length $2$. Let $\mN$ be a simple $2$-level network with $j \geq 1$ intermediate nodes. Let $[a_1,\dots,a_j]$ and $[b_1,\dots,b_j]^T$ be the matrix representation of the graph induced by the source, the intermediate nodes, and the terminal respectively where $a_i \in [a_1, \ldots, a_j]$ represents the number of incoming edges of intermediate node $V_i$ and $b_i \in [b_1, \ldots, b_j]$ represents the number of outgoing edges of intermediate node $V_i$. We denote $\mN$ by $([a_1,\dots,a_j],[b_1,\dots,b_j])$. In~\cite[Section~V]{beemer2023network} the authors introduced $5$ families of networks and a Singleton Cut-Set Bound~\cite[Corollary~6.2]{beemer2023network} that generalizes~\cite[Corollary 66]{ravagnani2018} and provide an upper bound on the capacity of the five families of networks introduced in~\cite[Section~V]{beemer2023network}. In Table~\ref{tab:my_label}, we summarize the upper bounds using~\cite[Corollary~6.2]{beemer2023network} and the best known upper bound on their capacities provided in~\cite[Section~VI and~VII]{beemer2023network} over large alphabet, that is, $|\alp| \gg t$.

\begin{table}[h!]
\renewcommand{\arraystretch}{1.6}
    \centering\resizebox{\columnwidth}{!}{
    \begin{tabular}{|c|c|c|c|}
    \hline 
        \hbox{Family} & \hbox{Upper bound} & \hbox{One-shot Capacity} & \hbox{Reference}\\
        \hline 
        $\mathfrak{A}_t = ([t,2t],[t,t]), t\geq 1$ & $C_1(\mathfrak{A}_t, \ad_{\mathfrak{A}_t}) \leq t$ & $C_1(\mathfrak{A}_t, \ad_{\mathfrak{A}_t}) < t$ &~\cite[Theorem 6.16]{beemer2023network}\\
        \hline 
        $\mathfrak{B}_s = ([1,s+1],[1,s]), s \geq 1$ & $C_1(\mathfrak{B}_s, \ad_{\mathfrak{B}_s}) \leq s$ & $C_1(\mathfrak{B}_s, \ad_{\mathfrak{B}_s}) < s$ &~\cite[Theorem 6.9]{beemer2023network}\\
        \hline 
       $\mathfrak{C}_t = ([t,t+1],[t,t]), t \geq 2$ & $C_1(\mathfrak{C}_t, \ad_{\mathfrak{C}_t}) \leq 1$ & $C_1(\mathfrak{C}_t, \ad_{\mathfrak{C}_t}) = 1$ &~\cite[Theorem 7.8]{beemer2023network}\\
       \hline
       $\mathfrak{D}_t = ([2t,2t],[1,1]), t \geq 1$ & $C_1(\mathfrak{D}_t, \ad_{\mathfrak{D}_t}) \leq 1$ & $C_1(\mathfrak{D}_t, \ad_{\mathfrak{D}_t}) = 1$ &~\cite[Theorem 7.11]{beemer2023network}\\
       \hline 
       $\mathfrak{E}_t = ([t,t+1],[1,1]), t \geq 1$ & $C_1(\mathfrak{E}_t, \ad_{\mathfrak{E}_t}) \leq 1$ & $C_1(\mathfrak{E}_t, \ad_{\mathfrak{E}_t}) < 1$ &~\cite[Theorem 6.15]{beemer2023network}\\
       \hline 
    \end{tabular}}
    \caption{One-shot Capacity of Families of Networks~\cite[Section~VI and~VII]{beemer2023network}.}
    \label{tab:my_label}
\end{table}

 The one-shot capacities of families $\mathfrak{A}_t, \mathfrak{B}_s,$ and $\mathfrak{E}_t$ have not been computed for large alphabets. Using~\cite[Theorem 1 and Lemma 5]{kurz2024capacity}, the capacity of $\mathfrak{B}_s$ for small alphabets is computed and presented in~\cite[Table 3]{kurz2024capacity}; the capacity computed is shown to be better than predicted by the conjecture on the capacity of $\mathfrak{B}_s$ in~\cite[Section~VI]{beemer2023network}.

\section{Multishot Capacity and Restricted Adversaries}\label{multishot1}

Our work focuses on the multishot capacity of an adversarial network. The capacity of adversarial networks in multiple rounds is linked to the $i$-th power channel. The $i$\textbf{-th power} of a channel $\Omega : \mathcal{X} \dashrightarrow \mathcal{Y}$ is the channel
$\Omega^i:\mathcal{X}^i \dashrightarrow \mathcal{Y}^i$
where $\Omega^i:= \Omega \times \dotsm \times \Omega$ is the $i$-fold product and $\Omega^i$ models $i$ uses of a network, see~\cite[Definition 10]{ravagnani2018}. A formal definition of the multishot capacity that extends~\cite[Definition 5]{beemer2023network} is provided below.
\begin{definition}[{\cite[Definition II.7]{multishotND}}]
    Let $i$ be a positive integer. The \textbf{$i$-shot capacity}~$C_i(\mN, \alp, \ad_{\mN})$ of an adversarial network $(\mN,\ad_{\mN})$ over an alphabet $\alp$ is the maximum $\alpha\in\R$ such that there exists an unambiguous code $C$ for $\Omega^i$ with $\smash{\alpha=\frac{\log_{|\mA|}(|C|)}{i}}$, meaning,
    \[C_i(\mN, \alp, \ad_{\mN}) = \max \left \{\alpha = \frac{\log_{|\mA|}(|C|)}{i}\,:\, C \hbox{ is an unambiguous code for each } \Omega^i \right\}.\]

\end{definition}
The multishot capacity ($i$-shot capacity) is an extension of the one-shot capacity in Definition \ref{oneshot}, where $i=1$ recovers the original definition.
Intuitively, it can be thought of as the maximum number of symbols that can be sent over $i$ transmission rounds without errors.  Throughout the paper, we consider the following two adversarial models. In each model, we assume that the adversary can attack up to $t$ edges chosen from the set of vulnerable edges. The models are the following:

\begin{enumerate}[label=A.\arabic*]
    \item\label{scenario1} The adversary attacks the same $t$ edges over $i$ uses of the network. 
    \item\label{scenario2} The adversary can change the $t$ edges to attack over $i$ uses of the network.
\end{enumerate}

\subsection{Multishot Capacity of the Diamond Network(s)}

We start by discussing the multishot capacity of the Diamond Network $\mD$ in Scenario~\ref{scenario1} and Scenario~\ref{scenario2}. The capacity of these networks was computed in~\cite{multishotND}.

\paragraph{Scenario~\ref{scenario1} for $\mD$.} In this scenario, the adversary can corrupt an edge in~$\{e_1,e_2,e_3\}$ of $\mD$ and cannot change the edge attacked. We start by noticing that using the strategy proposed in~\cite[Proposition III.1]{curious2021diamond}, one can easily show that $C_i(\mathcal{D}, \ad_{\mathcal{D}})~\geq~\textup{log}_{|\alp|}(|\alp|-1)$. In~\cite[Section~III]{multishotND}, we compute the $i$-shot capacity of $\mD$ and show that the previous bound is far from being achieved. In particular, the following holds.

\begin{proposition}\label{diam1} Let $\alp$ be an alphabet and $\ad_{\mD}$ be an adversary capable of corrupting an edge from the set $\{e_1,e_2,e_3\}$ of $\mD$. Then, the \textit{i-shot capacity} of $\mD$ in Scenario~\ref{scenario1} is 
\[\C_i(\mD,\alp, \ad_\mD)=\frac{\log_{|\mA|}(|\mA|^i-1)}{i}.\]
\end{proposition} 

An intuitive idea of why the $i$-shot capacity of the Diamond Network $\mD$ increases over multiple uses is that we can reserve a vector $(\star, \ldots, \star) \in \alp^i$, an extension of the strategy provided in~\cite[Section~3]{BEEMER202236}. It is clear that as $i$ grows, the larger the capacity, asymptotically going to $1.$ Thus, there is a gain in capacity over multiple uses of the Diamond Network~$(\mathcal{D}, \ad_{\mathcal{D}})$.

\paragraph{Scenario~\ref{scenario2} for $\mD$.}
In this scenario, the adversary $\ad_{\mD}$ can attack an edge in the set~$\{e_1,e_2,e_3\}$ and can change the edge attacked each transmission round. In~\cite[Section~IV]{multishotND}, we show that the value of the multishot capacity of $\mD$ is the same as the value of the one-shot capacity of $\mD.$ In particular, the following holds.

\begin{proposition}\label{diam2} Let $\alp$ be an alphabet and $\ad_{\mathcal{D}}$ be an adversary in Scenario~\ref{scenario2}. Then, the $i$-shot capacity of $\mD$ in Scenario~\ref{scenario2} is \[\C_i(\mD, \alp,\ad_\mD)=\log_{|\alp|}(|\mA|-1) = C_1(\mD, \alp, \ad_{\mD}).\]
\end{proposition}

Hence, over $i$ uses of $\mD$, we see that there is no gain in capacity. This result follows from the fact that the adversary is free to change the edge attacked each transmission round. Therefore,  the strategy provided in Scenario~\ref{scenario1} cannot be applied. The best strategy we can provide for $\mD$ in Scenario~\ref{scenario2} is reserving a symbol $\star$ and using the strategy provided in~\cite[Section~3]{BEEMER202236} $i$ times, which is drastically different from the result obtained in Scenario~\ref{scenario1}.\\

We now investigate the $i$-shot capacity of the Mirrored Diamond Network $\mS$ in Scenario \ref{scenario1} and Scenario \ref{scenario2}. 

\paragraph{Scenario \ref{scenario1} for $\mS$.} In this scenario, the adversary can attack an edge in the set of edges~$\{e_1,e_2,e_3,e_4\}$ and cannot change the edge attacked. Using~\cite[Proposition~12]{ravagnani2018}, it immediately follows that $C_i(\mS,\alp, \ad_{\mS}) \leq 1$. In~\cite[Section~IV]{multishotND} we proved that the value of the $i$-shot capacity of $\mS$ is the same as the one-shot capacity of $\mS$ indicating that using $\mS$ multiple times does not provide a gain in capacity in comparison to the one-shot capacity. In particular, the following holds.

\begin{proposition}\label{umd} Let $\alp$ be an alphabet and $\ad_{\mS}$ be an adversary able to corrupt up to one edge of the set of outgoing edges of the source of $\mS$. Then, the i-shot capacity of $\mS$ in Scenario \ref{scenario1} is \[C_i(\mS, \alp, \ad_{\mS}) = 1.\] 
\end{proposition}

Therefore, $C_i(\mS, \alp, \ad_{\mS}) = C_1(\mS, \alp, \ad_{\mS}) = 1$ and there is no advantage of using $\mS$ multiple times for communication.

\paragraph{Scenario \ref{scenario2} for $\mS$.} It is not hard to check that, for an adversary capable of attacking up to an edge and can change the edge attacked each transmission round, $C_i(\mS, \ad_{\mS}) = 1$ in Scenario~\ref{scenario2}. This follows by using arguments similar to the ones in the examples in~\cite[Section~IV]{multishotND}, as there is no assumption that the adversary may possibly corrupt the value of one edge of $\{e_1,e_2,e_3,e_4\}$ and the proof of~\cite[Proposition III.8]{multishotND}. The lower bound comes directly from~\cite[Proposition 12]{ravagnani2018}. Therefore, in scenarios \ref{scenario1} and \ref{scenario2}, we have~$C_i(\mS, \ad_{\mS}) = C_1(\mS, \ad_{\mS})$ and there is no gain in capacity of using $\mS$ more than once for communication. 

\subsection{Multishot Capacity of Families of Networks}

In this section, we compute the multishot capacity of families $\mathfrak{C}_t$, $\mathfrak{D}_t$ and $\mathfrak{E}_t$ in Scenario~\ref{scenario1} and Scenario~\ref{scenario2}. The goal is to determine for each family whether there is a gain in using them multiple times for communication versus using them once. For family $\mathfrak{E}_t$, we show that there is a gain in capacity of using $\mathfrak{E}_t$ in Scenario \ref{scenario1} and no gain in Scenario \ref{scenario2}. In contrast, for families $\mathfrak{C}_t$ and $\mathfrak{D}_t$, we show that there is no gain in using these networks multiple times for communication in both Scenarios \ref{scenario1} and \ref{scenario2}\\

\subsubsection{Multishot Capacity of Family $\mathfrak{C}_t$}

In this section, we compute the multishot capacity of Family $\mathfrak{C}_t$. In both scenarios \ref{scenario1} and \ref{scenario2}, we will prove that 
\[C_i(\mathfrak{C}_t, \alp, \ad_{\mathfrak{C}_t}) = 1\] rendering that $C_i(\mathfrak{C}_t, \alp, \ad_{\mathfrak{C}_t}) = C_1(\mathfrak{C}_t, \alp, \ad_{\mathfrak{C}_t})$, demonstrating that there is no gain in using Family $\mathfrak{C}_t$ multiple times for communication. Let $\mathcal{A}$ be an alphabet and let~$\ad_{\mathfrak{C}_t}$ be an adversary capable of corrupting $t$ edges from $\mU_S$, where $\mU_S$ is the set of edges incident with the source $S$. Let \[\Omega^i_{\mathfrak{C}_t}:= \Omega_{\mathfrak{C}_t}[\mathfrak{C}_t, \alp,\mathcal{F}, S \longrightarrow T, \mU_S, t] \times \dotsm \times \Omega_{\mathfrak{C}_t}[\mathfrak{C}_t,\alp,\mathcal{F}, S \longrightarrow T, \mU_S, t]\] be the $i$-th power channel for $\mathfrak{C}_t$. Let $\h_{\mathfrak{C}_t}$ be the channel that describes the action of the adversary where \[\h_{\mathfrak{C}_t} : \alp^{2t+1} \dashrightarrow \alp^{2t+1}  \quad\textup{ defined by}\quad \h_{\mathfrak{C}_t}(x) := \{y \in \mA^{2t+1} \,:\, d_{\h}(x,y) \leq t\}, \forall x \in \alp^{2t+1}.\] 

 It is easy to check that the largest unambiguous code for $\h_{\mathfrak{C}_t}$ is $|\mA|$. This implies that~$C_1(\h_{\mathfrak{C}_t}) = 1$. Next we describe a structural property that implies a code $C$ is unambiguous for $\h_{\mathfrak{C}_t}$. We have that a code $C$ is an unambiguous code for $\h_{\mathfrak{C}_t}^i$ if and only if for~$x,y \in C$ one of $d_{\h}(x^1,y^1) \geq 2t+1, \ldots, d_{\h}(x^i,y^i) \geq 2t+1$. We have the following remark.

\begin{remark}\label{conC} We provide a strategy for computing the size of the largest unambiguous code for~$\h_{\mathfrak{C}_t}^i$, which will be important in showing that $|C| = |\alp|^i$ in Proposition~\ref{ufc}. Assume $C$ is an unambiguous code for $\h^i_{\mathfrak{C}_t}$ with $|C| = |\mA|^i + 1$. Define \[x^1 = (x_1, \ldots, x_{2t+1}), \ldots, x^i = (x_1^i, \ldots, x^i_{2t+1}).\] With the approach that was applied~\cite[Example 9]{beemer2023network}, we claim that there are no two codewords that coincide in the first $|\mA|^i$ components. Let $c_1, \ldots, c_{|\mA|^i+1} \in C$. Without loss of generality, we can assume that for~$c_1, c_2 \in C$, $c_1^1 = c_2^1$ and $c_1 = 0$. This implies that~$c_3^1, \ldots, c_{|\mA|^i + 1}^1$ must have Hamming weight $2t+1$. If not, assume that $c_3^1$ has Hamming weight less than $2t+1$. This allows one of $\{c_1^2, c_2^2, c^2_3\}, \ldots ,\{c_1^i, c_2^i, c_3^i\}$ to be a unambiguous code of Hamming distance at least $2t+1$ and of cardinality $3$, contradicting the original capacity for~$\smash{\h_{\mathfrak{C}_t}}$. Since $\smash{c_3^1, \ldots, c_{|\mA|^i + 1}^1}$ have Hamming weight $2t+1$, a comparison of any two elements gives Hamming distance at most $t+1$. We assume that $C$ is unambiguous for~$\h_{\mathfrak{C}_t}^i$. Therefore,  one of $\smash{\{c_1^2,c_2^2,c_3^2\}, \ldots, \{c_1^i,c_2^i,c_3^i\}}$ has to be an unambiguous code for~$\h_{\mathfrak{C}_t}$ with cardinality $3$ and minimum Hamming distance $2t+1$ contradicting $C_1(\h_{\mathfrak{C}_t}) = 1$. 
\end{remark}

Finally, we are ready to compute the multishot capacity of $\mathfrak{C}_t$.

\begin{proposition}\label{ufc} Let $\alp$ be an alphabet. Recall that $\smash{\ad_{\mathfrak{C}_t}}$ can corrupt up to $t$ edges of the subset $\mU_S$. We have that the $i$-shot capacity of $\mathfrak{C}_t$ is  \[C_i(\mathfrak{C}_t, \alp, \ad_{\mathfrak{C}_t}) = 1.\] 
\end{proposition}
\begin{proof} We first show that $|C| \leq |\mA|^i$. The action the adversary may take is described by~$\h_{\mathfrak{C}_t}$. Let $\mathcal{F}$ be a network code for $\mathfrak{C}_t$. We will use $\mathfrak{C}_t$ in $i$ transmission rounds with network codes $\mathcal{F}_1, \ldots, \mathcal{F}_i$ respectively.  For $x = (x_1, \ldots, x_{(2t+1)i}) \in \mA^{2t+1} \times \dotsm \times \mA^{2t+1},$ we have that \[\Omega^i_{\mathfrak{C}_t}(x) := \h_{\mathfrak{C}_t}^i \left( \mathcal{F}^1_V(x_1,\ldots,x_{2t+1}), \ldots, \mathcal{F}^i(x_{(2t+1)(i-1)+1}, \ldots, x_{(2t+1)i} \right).\] Assume that there exists an unambiguous code $C$ with $|C| = |\mA|^i + 1$. Then, there exists an unambiguous code \[C' := \{\mF^1(x_1,\ldots,x_{2t+1}), \ldots, \mF^i(x_{(2t+1)(i-1)+1},\ldots,x_{(2t+1)i}))\} \subseteq  \alp^{2t+1} \times \dotsm \times \alp^{2t+1},\] which is unambiguous for $\h_{\mathfrak{C}_t}^i$ of cardinality $|\alp|^i + 1$. Since we have that $|C'| = (|\mA|)^i + 1,$ there must exist $y,y' \in C$, with $y \neq y'$ such that they coincide in the first $|\mA|^i$ components. This implies that~$C'$ is an unambiguous code for $\h_{\mathfrak{C}_t}^i$, which has two different codewords that coincide in the first $|\alp|^i$ components. However, it was shown in Remark \ref{conC}, such a code does not exist. Therefore,  $|C| \leq |\alp|^i$. 

It remains to show that $C_i(\mathfrak{C}_t, \alp, \ad_{\mathfrak{C}_t}) \geq 1$. This immediately follows from~\cite[Proposition~12]{ravagnani2018} using the channel $\Omega^i_{\mathfrak{C}_t}$ associated to $\mathfrak{C}_t$. We have that \[C_i(\Omega_{\mathfrak{C}_t}, \alp, \ad_{\mathfrak{C}_t}) = \frac{C_1(\Omega^i_{\mathfrak{C}_t}, \alp, \ad_{\mathfrak{C}_t})}{i} \geq \frac{i\; C_1(\Omega_{\mathfrak{C}_t}, \alp, \ad_{\mathfrak{C}_t})}{i} = C_1(\Omega_{\mathfrak{C}_t}, \alp, \ad_{\mathfrak{C}_t}) = 1,\]
as claimed.
\end{proof}

Notice that the proof of the multishot capacity of $\mathfrak{C_t}$ uses the same strategy as the proof of the multishot capacity of $\mS$. Moreover, there is no gain in using $\mathfrak{C}_t$ multiple times for communication in both scenarios.\\

\subsubsection{Multishot Capacity of $\mathfrak{D}_t$}

In this section, we compute the multishot capacity of Family $\mathfrak{D}_t$. In both scenarios \ref{scenario1} and \ref{scenario2}, we will prove that 
\[C_i(\mathfrak{D}_t, \alp, \ad_{\mathfrak{D}_t}) = 1.\]

Suppose that an adversary $\ad_{\mathfrak{D}_t}$ can corrupt up to $t$ edges of $\mU_S$, where $\mU_S$ is the set of edges incident with the source $S$, and may change the edges attacked each transmission round. Let $\mathcal{A}$ be an alphabet and let \[\Omega^i_{\mathfrak{D}_t}:= \Omega_{\mathfrak{D}_t}[\mathfrak{D}_t,\alp,\mathcal{F}, S \longrightarrow T, \mU_S, t] \times \dotsm \times \Omega_{\mathfrak{D}_t}[\mathfrak{D}_t,\alp,\mathcal{F}, S \longrightarrow T, \mU_S, t]\] be the $i$-th power channel for $\mathfrak{D}_t$ that represents $i$ uses of $\mathfrak{D}_t$. Let $\h_{\mathfrak{D}_t}$ be the channel that describes the action of the adversary with \[\h_{\mathfrak{D}_t} : \alp^{4t} \dashrightarrow \alp^{4t} \quad\textup{ defined by}\quad \h_{\mathfrak{D}_t}(x) := \{y \in \mA^{4t}\,:\, d_{\h}(x,y) \leq t\},\] for any $x\in \alp^{4t}$. We note that the largest unambiguous code for $\h_{\mathfrak{D}_t}$ is $|\mA|$ and there is no larger code. This implies that $C_1(\h_{\mathfrak{D}_t}) = 1$. Next, we describe a structural property that implies a code $C$ is unambiguous for $\h_{\mathfrak{D}_t}$. We find that a code $C$ is unambiguous for $\h_{\mathfrak{D}_t}^i$ if and only if one of~$d_{\h}(x^1,y^1) \geq 3t, \ldots, d_{\h}(x^i,y^i) \geq 3t$. We have the following remark.

\begin{remark}\label{conD} We provide a strategy for the largest unambiguous code for $\mathfrak{D}_t$ and show that it is $|C| = |\mA|^i$. We let $C$ be any unambiguous code for $\h^i_{\mathfrak{D}_t}$ with $|C| = |\mA|^i + 1$. Define~$x^1~=~(x_1, \ldots, x_{4t}), \ldots, x^i = (x_1^i, \ldots, x^i_{4t})$. With the approach that was applied in Remark \ref{conC} and~\cite[Example 9]{ravagnani2018}, we wish to show that there are no two codewords that coincide in the first $|\mA|^i$ components. Let $c_1, \ldots, c_{|\mA|^i+1} \in C$. Without loss of generality, we can assume that for $c_1, c_2 \in C$, $c_1^1 = c_2^1$ and $c_1 = 0$. Notice that the proof is the same as in Remark \ref{conC} with the exception that $\smash{c_3^1, \ldots, c_{|\mA|^i + 1}^1}$ must have Hamming weight $3t$. If this is not the case and we assume $C$ is unambiguous for $\h_{\mathfrak{D}_t}^i$, one can find an unambiguous code for $\h_{\mathfrak{D}_t}$ with cardinality $3$ and minimum Hamming distance $3t$ contradicting the capacity of $\h_{\mathfrak{D}_t}$.

\end{remark}

We are now ready to compute the multishot capacity of Family $\mathfrak{D}_t$ in both scenarios.

\begin{proposition}\label{capacityD} Let $\alp$ be an alphabet and recall that the adversary $\ad_{\mathfrak{D}_t}$ can corrupt up to $t$ edges of $\mU_S$. Then, the i-shot capacity of $\mathfrak{D}_t$ is \[C_i(\mathfrak{D}_t, \alp,\ad_{\mathfrak{D}_t}) = 1.\] 
\end{proposition}

The proof of this result is similar to the proof of Family $\mathfrak{C}_t$ in Scenario~\ref{scenario1} using Remark~\ref{conD}.  Therefore, there is no gain in using $\mathfrak{D}_t$ multiple times for communication. Notice that the proof does not depend on whether the adversary does or does not change the edges attacked each transmission round and therefore applies in both scenarios. Thus, we have that~$C_i(\mathfrak{C}_t, \alp, \ad_{\mathfrak{C}_t}) = C_i(\mathfrak{D}_t, \alp, \ad_{\mathfrak{D}_t}) = 1$ in Scenario \ref{scenario2}.

This shows there is no strategy for Family $\mathfrak{C}_t$ and Family $\mathfrak{D}_t$ that provides a gain in capacity in Scenario \ref{scenario1}.

\subsubsection{Multishot Capacity of $\mathfrak{E}_t$}

In this section, we compute the multishot capacity of Family $\mathfrak{E}_t$. In Scenario~\ref{scenario1}, we will prove that 

\[C_i(\mathfrak{E}_t, \alp, \ad_{\mathfrak{E}_t}) = \frac{\log_{|\alp|}(|\alp|^i - b)}{i},\] showing a gain in capacity over $i$ uses of $\mathfrak{E}_t$. In contrast, we will prove that in Scenario~\ref{scenario2}, 
\[C_i(\mathfrak{E}_t, \alp, \ad_{\mathfrak{E}_t}) = \log_{|\alp|}(|\alp| - b),\] rendering no gain in using $\mathfrak{E}_t$ over $i$ transmission rounds. 

The results for Family $\mathfrak{E}_t$ suggest that calculating the multishot capacity, devising a capacity-achieving strategy for an adversarial network, and determining whether there is an increase in capacity over multiple uses is feasible without prior knowledge of the one-shot capacity.

\paragraph{Scenario~\ref{scenario1} for $\mathfrak{E}_t$.} Assume that an adversary $\ad_{\mathfrak{E}_t}$ can attack up to $t$ edges of the subset $\mU_s,$ where $\mathcal{U}_S$ is the set of edges incident with the source $S$, but cannot change the edges attacked each transmission round. In~\cite[Theorem 6.15]{beemer2023network}, the authors proved that \[C_1(\mathfrak{E}_t, \alp, \ad_{\mathfrak{E}_t}) < 1.\] 

Let $B$ be a set of reserved vectors in $\alp^{2t+1}$ with $|B| = b$. Let $C \subseteq \alp^{2t+1}$ be the largest unambiguous code for $\mathfrak{E_t}$ for $\Omega[\mathfrak{E}_t, \mA, \mF, \textup{out}(S)\to \textup{in}(T), \mathcal{U}_S,t]$, with 
\begin{equation}\label{eq:Cb}
    |C| = |\alp| - b.
\end{equation}
Suppose that there exists a capacity-achieving pair $(C,\mF)$. We have that \[C_1(\mathfrak{E}_t, \alp, \ad_{\mathfrak{E}_t}) = \log_{|\mA|}(|\mA| - b).\] The goal is to show that \begin{equation*}
     C_i(\mathfrak{E}_t, \alp, \ad_{\mathfrak{E}_t}) = \frac{\log_{|\mA|}(|\mA|^i - b)}{i}.
\end{equation*}
We define the sets
\begin{align*}
    B^i&:=\{\underbrace{(k \mid \dots \mid k)}_{i\textup{-times}} \,:\, k \in B\},\\
    B_1^i&:=\{(k_1,\ldots,k_t \mid \dots \mid k_1,\ldots,k_t) \,:\, k \in B\},\\
    B_2^i&:=\{(k_{t+1},\ldots,k_{2t+1}\mid  \dots \mid k_{t+1},\ldots,k_{2t+1}) \,:\, k \in B\},
\end{align*}
where ``$\,\mid\,$'' denotes the vector concatenation. Note that $|B^i| = |B_1^i| = |B_2^i| = b$. Let $C \subseteq \mA^{a_1 + a_2 +\dots+a_n}$. For a given $x \in C$, we denote by $B_t^{\h}(x)$ the Hamming ball of radius $t$ with center $x$ and $S_t^{\h}(x)$ the shell of that ball, meaning,
\begin{equation*}
B_t^{\h}(x) = \{y \in \mA^{a_1 + a_2 +\dots+a_n}\,:\, d_{\h}(x,y) \leq t\}, S_t^{\h}(x) = \{y \in B_t^{\h}(x) \,:\, d_{\h}(x, y) = t\}.
\end{equation*} Let  $\pi_k^i : \mA^{(2t+1)i} \to \mA^i$ be the projection onto the coordinates corresponding to the edges to intermediate node $V_k$ of the simple $2$-level network $\mN.$ The following result provides a construction of the maximum unambiguous code for $\mathfrak{E}_t$ with cardinality $|\alp|^i-b$.

\begin{proposition}\label{boundE} Let $\alp$ be an alphabet, $\mU_S$ be the set of edges incident with the source~$S$ and let $\ad_{\mathfrak{E}_t}$ be an adversary able to corrupt up to $t$ edges of $\mU_S$ of $\mathfrak{E}_t$. Let $b$ be as in~\eqref{eq:Cb}. Then, the $i$-shot capacity of $\mathfrak{E}_t$ in Scenario \ref{scenario1} is \[C_i(\mathfrak{E}_t, \alp, \ad_{\mathfrak{E}_t}) = \frac{\log_{|\mA|}(|\mA|^i - b)}{i}.\] 
\end{proposition}
\begin{proof}
We start with the lower bound. Our goal is to show that the code 
    \begin{equation*}
        C = \{(a\mid \dots \mid a) \in (\mA^i)^{2t+1} \,:\, (a_1, \ldots,a_{(2t+1)i}) \notin B^i\} \subseteq (\alp^i)^{2t+1}
    \end{equation*}
    is unambiguous for $\Omega^i[\mathfrak{E}_t, \mA, \mF, \textup{out}(S)\to \textup{in}(T), \mathcal{U}_S,t]$, for a certain network code $\mF$. Suppose that the adversary changes $t$ symbols in the network. We consider two cases: \\ 
    
 \textbf{Case 1:} If $V_1$ receives a vector of the form $(b_1 \mid \dots \mid b_j \mid a \mid \dots \mid a)$ with~$j \geq t/2$, then the vertex $V_1$ forwards a reserved symbol $k \in B_1^i$ and the vertex $V_2$ forwards the received symbol $a \in \alp^{(t+1)i}\setminus B_2^{i}$. Therefore, terminal receives an element of 
    \begin{equation*}
        \Omega^i[\mathfrak{E}_t, \mA, \mF, S \longrightarrow T, \mathcal{U}_S,t](a \mid \dots \mid a) = \{(k\mid a)\, : \, a \in \mA^{(t+1)i} \setminus B_2^i, \, k \in B_1^i\}
    \end{equation*} and trusts the vector from $V_2$.\\

     \textbf{Case 2:} If $V_2$ receives a vector of the form $(b_1\mid \dots \mid b_l\mid a \mid \dots \mid a)$ with~$l \geq t/2 + 1$, then the vertex $V_1$ forwards the received symbol $a \in \alp^{ti}\setminus B_1^{i}$ and the vertex $V_2$ forwards a reserved symbol $k \in B_2^{i}$. Therefore, terminal receives an element of  
    \begin{equation*}
         \Omega^i[\mathfrak{E}_t, \mA, \mF, S \longrightarrow T, \mathcal{U}_S,t](a \mid \dots \mid a) = \{(a\mid k) \, : \,
         a \in \mA^{ti} \setminus B_1^i, \, k \in B_2^i\}
    \end{equation*} and trusts the vector from $V_1$. It follows that for any $c,c' \in C$, we have that
    \begin{equation*}
        \Omega^i[\mathfrak{E}_t, \mA, \mF, S \longrightarrow T, \mathcal{U}_S,t](c)\cap \Omega^i[\mathfrak{E}_t, \mA, \mF, S \longrightarrow T, \mathcal{U}_S,t](c') \neq \emptyset 
    \end{equation*} 
    if and only if $c=c'$. This shows that $C$ is an unambiguous code. One can easily check that $|C| \geq |\mA|^i - b.$\\

 It remains to show the upper bound. The following generalizes~\cite[Theorem 6.15]{beemer2023network} to the multishot setting. Assume towards a contradiction that there exists an unambiguous code $C$ for $\Omega^i[\mathfrak{E}_t, \mA, \mF, S \longrightarrow T, \mathcal{U}_S,t]$ such that $|C| = |\mA|^i - b + 1$ for some choice of network code $\mF = {\mF_1,\mF_2}$. Let $x = (x_1, \ldots, x_{(2t+1)i})$ where $$x_1 = (x_1, \ldots, x_{2t+1}), \ldots, x_i = (x_{(2t+1)(i-1) + 1}, \ldots, x_{(2t+1)i}).$$ Since $C$ is an unambiguous code, it must be the case that at least one of the following hold:~$d_H(x_1,y_1)\geq 2t+1, \ldots, d_H(x_i,y_i)\geq 2t+1$. Then, $\pi_1^i(C)$ contains~$|\mA|^i-b + 1$ distinct elements. Since the restriction of $\mF_1$ to $\pi_1^i(C)$ is injective, it follows that $$\mF_1(\pi_1^i(C)) = (\mA^i \setminus B^i)\cup \{l\},$$ for some $l \in B^i$. To conclude, there must exist $x,y \in C$, with $x \neq y$ and some $e \in \mA^{(2t+1)i}$ such that \[\mF_1(\pi_1^i(x)) = \mF_1(\pi_1^i(e))\hspace{3ex} \hbox{and} \hspace{3ex} d_H(\pi_1^i(e),\pi_1^i(y)) \leq ti-1.\] We then have that 
\begin{equation*}
    x' := (x_1, \ldots, x_{ti}, x_{ti+1}, y_{ti+2}, \ldots, y_{(2t+1)i}) \in B_t^H(x), 
\end{equation*}
\begin{equation*}
    y' := (e_1, \ldots, e_{ti},x_{ti+1}, y_{ti+2}, \ldots, y_{(2t+1)i}) \in B_t^H(y).
\end{equation*}
Lastly, we can observe that \[(\mF_1(\pi_1^i(x'),\mF_2(\pi_2^i(x')) = (\mF_1(\pi_1^i(y'),\mF_2(\pi_2^i(y')) \in \Omega^i(x)\cap \Omega^i(y),\] which contradicts $C$ being an unambiguous code. Therefore,  \begin{equation*}
    C_i(\mathfrak{E}_t, \alp, \ad_{\mathfrak{E}_t}) \leq \frac{\log_{|\mA|}(|\mA|^i - b)}{i}. \qedhere
\end{equation*}
\end{proof}

The above theorem shows a gain in using $\mathfrak{E}_t$ multiple times for communication without having to compute the one-shot capacity. This result opens the discussion of multishot capacities of adversarial networks without an explicit strategy for the one-shot capacity. 
Note that if $t=1$, we recover the Diamond network $\mathcal{D}$ where the multishot capacity was computed as  $\smash{C_i(\mD, \alp, \ad_{\mD}) = \frac{\log_{|\mA|}(|\mA|^i-1)}{i}}$. Also, $b=1$ for $\mD$ and in Scenario~\ref{scenario1}, there is a gain in using $\mD$ multiple times for communication.

\paragraph{Scenario \ref{scenario2} for $\mathfrak{E}_t$.} Recall that in Scenario~\ref{scenario2}, the adversary $\ad_{\mathfrak{E}_t}$ can corrupt up to~$t$ edges of the edges outgoing the source of $\mathfrak{E}_t$ and is free to change the edges attacked.

Let $B$ be the set of reserved vectors as in Scenario~\ref{scenario1}. Assume that there exists a capacity achieving pair $(C,\mF)$ such that $C_1(\mathfrak{E}_t, \alp, \ad_{\mathfrak{E}_t}) =\log_{|\mA|}(|\mA|-b).$ We can immediately say that $C_i(\mathfrak{E}_t, \alp, \ad_{\mathfrak{E}_t}) \geq \log_{|\mA|}(|\mA|-b)$ by~\cite[Proposition~12]{ravagnani2018}. We wish to show that \[C_i(\mathfrak{E}_t, \alp,  \ad_{\mathfrak{E}_t}) = \log_{|\alp|}(|\alp| - b),\] rendering that $C_1(\mathfrak{E}_t, \alp, \ad_{\mathfrak{E}_t}) = C_i(\mathfrak{E}_t, \alp, \ad_{\mathfrak{E}_t}).$

Suppose that an adversary $\ad_{\mathfrak{E}_t}$ can corrupt up to $t$ edges of $\mathfrak{E}_t$. The action of the adversary can be represented as the channel \[\textup{H}_{\mathfrak{E}_t}: \alp^{2t+1} \dashrightarrow \alp^{2t+1}\quad \textup{defined by}\quad \textup{H}_{\mathfrak{E}_t}(x) = \{y \in \alp^{2t+1} \,:\, d_{\h}(x,y) \leq t\}.\] 

for all $x\in\alp^{2t+1}$. From the construction of the unambiguous code for $C_1(\mathfrak{E}_t, \alp, \ad_{\mathfrak{E}_t})$, the largest unambiguous code for $\textup{H}_{\mathfrak{E}_t}$ has cardinality $|\alp| - b$ and there is no larger code (since $b$ vectors are reserved). Thus, $\smash{C_1(\textup{H}_{\mathfrak{E}_t}) = \log_{|\alp|}(|\alp| - b)}$. We wish to show that there does not exist an unambiguous code $C$ for $(\mathfrak{E}_t, \ad_{\mathfrak{E}_t})$ where $|C| = (|\alp|-b)^i + 1$, that is,~$|C|~\leq~(|\alp|-b)^i$. We provide a structural property that shows $C$ is a unambiguous code for~$\h_{\mathfrak{E}_t}$. We also have that $C$ is unambiguous for $\h_{\mathfrak{E}_t}^i$ if and only if at least one of the following holds: $d_{\h}(x^1,y^1) \geq 2t+1, \ldots, d_{\h}(x^i,y^i) \geq 2t+1.$ We have the following remark.

\begin{remark}\label{ufamE} We provide a strategy for computing the largest unambiguous code for~$\mathfrak{E}_t$ in multiple transmission rounds. Let $C$ be an unambiguous code for $\h_{\mathfrak{E}_t}^i$ where we have that $|C| = (|\alp|-b)^i + 1$. Define $x^1 = (x_1, \ldots, x_{2t+1}), \ldots, x^i = (x_1^i, \ldots, x^i_{2t+1})$. With the approach that was applied in~\cite[Example~9]{ravagnani2018}, we claim that there are no two codewords that coincide in the first~$(|\mA| - b)^i$ components. Let $\smash{c_1, \ldots, c_{(|\alp|-b)^i + 1} \in C}$. Without loss of generality, we can assume that for $\smash{c_1, c_2 \in C}$, $c_1^1 = c_2^1$ and $c_1 = 0$. This implies that $\smash{c_3^1, \ldots, c_{(|\alp|-b)^i + 1}^1}$ must have Hamming weight~$2t+1$. If not, assume that $c_3^1$ has Hamming weight less than~$2t+1$. This allows one of $\{c_1^2, c_2^2, c^2_3\}, \ldots ,\{c_1^i, c_2^i, c_3^i\}$ to be an unambiguous code of Hamming distance at least $2t+1$ and of cardinality $3$, contradicting the original capacity for $\h_{\mathfrak{E}_t}$. Since $c_3^1, \ldots, c_{(|\alp|-b)^i + 1}^1$ have Hamming weight $2t+1$, a comparison of any two elements gives Hamming distance at most $t+1$. We assume that $C$ is unambiguous for $\h_{\mathfrak{C}}^i$, therefore, one of~$\{c_1^2,c_2^2,c_3^2\}, \ldots, \{c_1^i,c_2^i,c_3^i\}$ has to be an unambiguous code for $\h_{\mathfrak{C}}$ with cardinality $3$ and minimum Hamming distance $2t+1$ contradicting $C_1(\h_{\mathfrak{C}}) = \log_{|\alp|}(|\alp| - b)$. 
\end{remark}

The next proposition computes the multishot capacity~$C_i(\mathfrak{E}_t, \alp, \ad_{\mathfrak{E}_t})$ in Scenario~\ref{scenario2}. 

\begin{proposition}\label{boundE2}  Let $\alp$ be an alphabet, $\mU_S$ be the set of sources connected to the source~$S$, and  $\ad_{\mathfrak{E}_t}$ be an adversary able to corrupt up to $t$ edges of $\mU_S$ of $\mathfrak{E}_t$. Then, the $i$-shot capacity of $\mathfrak{E}_t$ in Scenario \ref{scenario2} is \[C_i(\mathfrak{E}_t, \alp, \ad_{\mathfrak{E}_t}) = \log_{|\alp|}(|\alp|-b).\]
\end{proposition}
\begin{proof}
Assume that we use a $2t+1$-times repetition code~$C=\{(a,\ldots, a): a \in \alp\}$ in each round. We first show that $|C| \leq (|\alp|-b)^i$. Let $\ad_{\mathfrak{E}_t}$ be an adversary that is restricted to corrupting $t$ edges of the network $\mathfrak{E}_t$ on the first level. The action of the adversary is given by $\h_{\mathfrak{E}_t}$  as described above. Let $\mathcal{F}$ be a network code for $\mathfrak{E}_t$. We will use the same network code each transmission round. We have that $i$ uses of $\mathfrak{E}_t$ is represented by $\Omega^i_{\mathfrak{E}_t}$. For $x = (x_1, \ldots, x_{(2t+1)i}) \in \mA^{2t+1} \times \dotsm \times \mA^{2t+1},$ we have that \[\Omega^i_{\mathfrak{E}_t}(x) := \h_{\mathfrak{E}_t}^i(\mathcal{F}^1_V(x_1,\ldots,x_{2t+1}), \ldots, \mathcal{F}^i_V(x_{(2t+1)(i-1)+1}, \ldots, x_{(2t+1)i}).\] Now assume that there exists a unambiguous code $C$ with $|C| = (|\alp|-b)^i + 1$. Then there exists a unambiguous code \[C' := \{\mF^1(x_1,\ldots,x_{2t+1}), \ldots, \mF^i(x_{(2t+1)(i-1)+1},\ldots,x_{(2t+1)i}))\} \subseteq  \alp^{2t+1} \times \dotsm \times \alp^{2t+1}\] of cardinality $(|\alp| - b)^i + 1$ which is unambiguous for $\h_{\mathfrak{E}_t}^i$. Therefore, there must exist~$y,y' \in C$, with~$y \neq y'$, such that they coincide in the first $(|\mA|-b)^i$ components. We have that~$C'$ is a unambiguous code for $\h_{\mathfrak{E}}^i$ which has two different codewords that coincide in the first $(|\alp|-b)^i$ components. However, it was shown in Remark~\ref{ufamE}, such a code does not exist. Therefore,~$|C| \leq (|\alp|-b)^i$.

The lower bound immediately follows from~\cite[Proposition 12]{ravagnani2018} using the $i-$th power channel $\Omega_{\mathfrak{E}_t}^i.$
\end{proof}

The above proposition shows that there is no gain in using $\mathfrak{E}_t$ multiple times for communication in this scenario.

\begin{remark} The upper bound comes directly from the proof of the upper bounds provided for $\mathcal{D},\mathcal{S},\mathfrak{C}_t$ and $\mathfrak{D}_t$ by assuming there exists unambiguous code $C$ where we have that $|C| = (|\alp|-b)^i + 1$ for $\h^i_{\mathfrak{E}_t}$ and showing that a code of this size contradicts the fact that there are no two codewords that coincide in the first $(|\alp|-b)^i$ components. Hence,  $C_i(\mathfrak{E}_t, \alp, \ad_{\mathfrak{E}_t}) = \log_{|\alp|}(|\alp| - b)$ and $C_1(\mathfrak{E}_t, \alp, \ad_{\mathfrak{E}_t}) = C_i(\mathfrak{E}_t, \alp, \ad_{\mathfrak{E}_t})$. Thus, there is no gain of using $\mathfrak{E}_t$ multiple times for communication in this scenario.  The results for $\mathfrak{E}_t$ follow similarly to the results of $\mD$, which is to be expected since $\mathfrak{E}_1 = \mD.$
\end{remark}

\section{Multishot Double Cut-Set Bound}\label{doubleM}
In~\cite{beemer2023network}, a reduction from $3$-level networks to $2$-level networks was provided, bounding the capacity of a $3$-level network by the $2$-level network associated with it. In this section, we extend this reduction to the multishot regime, bounding the multishot capacity of simple $3$-level networks to the simple $2$-level networks associated to them. We also provide a multishot version of the Double Cut-Set Bound~\cite[Theorem 8.6]{beemer2023network} and a multishot version of the Singleton Cut-Set Bound~\cite[Corollary 66]{ravagnani2018}. Using the procedure described in~\cite[Section~V]{beemer2023network}, we will be able to upper bound the multishot capacity of an arbitrary network~$\mN$ with that of an induced simple $2$-level network.

We start by describing the transfer of information from $\mE_1$ to $\mE_2$, where $\mE_1, \mE_2$ are edge cuts. We say edge cut $\mE_1$ precedes edge cut $\mE_2$, denoted $\mE_1 \preccurlyeq \mE_2$, if there exists a directed path in $\mN$ from~$\mE_1$ to~$\mE_2$. Let $\mN$ be a network, $\mE_1$ and $\mE_2$ be edge cuts such that~$\mE_1 \preccurlyeq \mE_2$. Let $\alp$ be a network alphabet and $\mF$ a network code for $(\mN,\alp)$, $\mU \subseteq \mE$ be a set of edges, and $t \ge 0$. Then we denote by 
\begin{equation} \label{chtd}
\Omega[\mN,\alp,\mF,\mE_1 \longrightarrow \mE_2,\mU \cap \mE_1,t]: \alp^{|\mE_1|} \dashrightarrow \alp^{|\mE_2|}
\end{equation}
the channel that describes the transfer from the edges of $\mE_1$ to those of $\mE_2$, when an adversary can corrupt up to $t$ edges from $\mU \cap \mE_1$. This notation allows us to focus on the transfer of information between the two edge cuts $\mE_1$ and $\mE_2$, instead of looking at the transfer of information from the source~$S$ to the set of terminals $\bold{T}.$ We begin with the following proposition that extends~\cite[Theorem~5.8]{beemer2023network} to the multishot setting.

\begin{proposition}Let $\mathcal{N}_3$ be a simple $3$-level network, and let $\mathcal{N}_2$ be the simple $2$-level network associated to it. Define $\mathcal{U}_3$ and $\mathcal{U}_2$ to be the set of edges incident with the sources of $\mathcal{N}_3$ and $\mathcal{N}_2$ and define $\ad_{\mN_3}, \ad_{\mN_2}$ as adversaries able to corrupt up to $t$ edges of~$\mU_3$ and~$\mU_2$  respectively. Then for all alphabets $\alp$, \begin{equation}
    C_i(\mathcal{N}_3, \alp, \ad_{\mN_3}) \leq C_i(\mathcal{N}_2, \alp, \ad_{\mN_2}).
\end{equation}
\end{proposition}
\begin{proof} Let $\mF_3$ be a network code and $C_3$ be an outer code for an alphabet $\alp$ and the simple $3$-level network $\mN_3$ which is unambiguous for $\Omega[\mN_3,\alp,\mF_3,S \longrightarrow T, \mU_3, t]$, for all~$T \in \bold{T}$. We note that the same network code is used each transmission round. Let~$C_3'$ be an unambiguous code for the channel $\Omega^i[\mN_3,\alp,\mF_3,S \longrightarrow T, \mU_3, t]$, where $\Omega^i$ is the channel associated to~$i$ uses of~$\mN_3$. We wish to show that there exists~$\mF_2$ such that~$C_3'$ is unambiguous for $\Omega^i[\mN_2,\alp,\mF_2,S \longrightarrow T, \mU_2, t],$ for all $T \in \bold{T}.$ Let $\smash{G^{3,2}, V_{ij}^3, \mF_{V_{ijk}^3}}$, the neighborhood of~$V_{ij}^3$, $\mF_2$ and $\mF_{V_i}$ be defined as in~\cite[Theorem 5.8]{beemer2023network}. We have that each intermediate node $V_i$ in $\mN_2$ corresponds to connected component $i$ of $G^{3,2}$ and is a composition of functions at nodes in $V_1$ and $V_2$ of $\mN_3$. It was shown that the fan out set~$\Omega(x)$ of any $x \in C_3$ over the channel $\Omega[\mN_3,\alp,\mF_3,S \longrightarrow T, \mU_3, t]$ is exactly equal to the fan-out set~$\Omega(x)$ of~$x$ over $\Omega[\mN_2,\alp,\mF_2,S \longrightarrow T, \mU_2, t]$. Therefore,  it is easy to check that this is the case for~$\Omega^i[\mN_3,\alp,\mF_3,S \longrightarrow T, \mU_3, t]$ and $\Omega^i[\mN_2,\alp,\mF_2,S \longrightarrow T, \mU_2, t]$, by definition of~$\Omega^i$. We can conclude that $C_3'$ is an unambiguous code for the channel~$\Omega^i[\mN_2,\alp,\mF_2,S \longrightarrow T, \mU_2, t]$. Therefore, we have that \[C_i(\mathcal{N}_3, \alp, \ad_{\mN_3}) \leq C_i(\mathcal{N}_2, \alp,  \ad_{\mN_2}),\] as desired. 
\end{proof}  

The earlier proposition indicates that by understanding the relationship between a simple $2$-level network and a simple 3-level network, we can determine the multishot capacity of a larger network by calculating the multishot capacity of a smaller one. Before proving the main results, we have the following remark.

\begin{remark} \label{rmk:imm}
Note that as in~\cite{beemer2023network}, we do not require the edge cuts $\mE_1$ and $\mE_2$ to be minimal or even  \textit{antichain} cuts (i.e., cuts where any two different edges cannot be compared with respect to the order $\preccurlyeq$). Furthermore, $\Omega^i[\mN,\mA,\mF,\mE_1 \longrightarrow \mE_2,\mU \cap \mE_1,t]$ inherits from the channel~$\Omega[\mN,\mA,\mF,\mE_1 \longrightarrow \mE_2,\mU \cap \mE_1,t]$ the property that the predecessors are considered first in the network topology. Therefore,~$\Omega^i[\mN,\mA,\mF,\mE_1 \longrightarrow \mE_2,\mU \cap \mE_1,t]$
provides the value of each edge of the edge cut $\mE_2$ as a \textit{function} of the values of the immediate predecessors in $\mE_1$. 
\end{remark}

We now illustrate and extend the results on $2$-level and $3$-level networks derived in~\cite[Section~8]{beemer2023network} and discuss how they can be combined and applied to study the multishot capacity of an arbitrarily large and complex network $\mN$. The following result generalizes the Double Cut-Set Bound~\cite[Theorem 8.6]{beemer2023network} to the multishot scenario.

\begin{theorem}[The Multishot Double Cut-Set Bound]\label{dcsb} Let $\mN$ be a network, $\alp$ be an alphabet and $\mathcal{U} \subseteq \mathcal{E}$ be a set of edges. Let $\mathcal{E}_1$ and $\mathcal{E}_2$ be edge cuts such that $\mathcal{E}_1 \preccurlyeq \mathcal{E}_2$, let~$T \in \bold{T}$ and let $\mathcal{F}$ be a network code. Then, the one-shot capacity of $\mN$ satisfies \[C_1(\Omega^i[\mathcal{N},\alp, \mathcal{F}, S \longrightarrow T, \mathcal{U},t]) \leq \max_{\mathcal{F}} C_1(\Omega^i[\mathcal{N},\alp, \mathcal{F},\mathcal{E}_1 \longrightarrow \mathcal{E}_2, \mathcal{U},t])\] with the maximum taken over all network codes $\mathcal{F}$ for $(\mathcal{N},\alp)$. 
\end{theorem}

\begin{proof} 
Let $\mathcal{F}$ be a network code for $(\mathcal{N},\alp)$. Using a similar approach as in~\cite[Theorem~8.6]{beemer2023network}, consider the scenario where an adversary $\ad_{\mN}$ can corrupt up to $t$ edges of~$\mathcal{U}\cap \mathcal{E}_1$. This scenario is modeled by the channel \[\Omega^* := \Omega[\mathcal{N},\alp,\mathcal{F},\textup{out}(S) \longrightarrow \mathcal{E}_1, \mathcal{U}\cap\textup{out}(S) ,0] \concat \Omega[\mathcal{N},\alp,\mathcal{F},\mathcal{E}_1 \longrightarrow \mathcal{E}_2, \mathcal{U}\cap \mathcal{E}_1 ,t] \concat\] \[\Omega[\mathcal{N},\alp,\mathcal{F},\mathcal{E}_2 \longrightarrow \textup{in}(T), \mathcal{U}\cap \mathcal{E}_2 ,0].\] We are interested in $i$ uses for the channel, therefore we now consider the power channel~$(\Omega^*)^i.$ Using~\cite[Proposition 18]{ravagnani2018}, we have 
\begin{multline*}
    (\Omega^*)^i = \Omega^i[\mathcal{N},\alp,\mathcal{F},\textup{out}(S) \longrightarrow \mathcal{E}_1, \mathcal{U}\cap \textup{out}(S) ,0]) \blacktriangleright \\ (\Omega^i[\mathcal{N},\alp,\mathcal{F},\mathcal{E}_1 \longrightarrow \mathcal{E}_2, \mathcal{U}\cap \mathcal{E}_1 ,t])  \blacktriangleright  \Omega^i[\mathcal{N},\alp,\mathcal{F},\mathcal{E}_2 \longrightarrow \textup{in}(T), \mathcal{U}\cap \mathcal{E}_2 ,0]).
\end{multline*}
We notice that the former and latter channels are deterministic (see \cite[Definition 3.1]{beemer2023network}) since we consider an adversarial power of 0. These channels describe the
transfer from the source to $\mE_1$ and from $\mE_2$ to $T$, respectively. For ease of notation, we let 

\[
\Omega^i_1 := \Omega^i[\mathcal{N},\alp,\mathcal{F},\textup{out}(S) \longrightarrow \mathcal{E}_1, \mathcal{U}\cap \textup{out}(S) ,0]\] and 
\[\Omega^i_2 := \Omega^i[\mathcal{N},\alp,\mathcal{F},\mathcal{E}_2 \longrightarrow \textup{in}(T), \mU \cap \mE_2, 0].
\]
We note that by Definition \ref{coarser}, the channel $\Omega^i[\mathcal{N},\alp,\mathcal{F},S \longrightarrow T,\mathcal{U},t]$ is coarser than the channel $\Omega_1^i \blacktriangleright \Omega^i[\mathcal{N},\alp,\mathcal{F},\mathcal{E}_1 \longrightarrow \mathcal{E}_2, \mathcal{U}\cap \mathcal{E}_1 ,t] \blacktriangleright \Omega^i_2$ since in the latter channel, the errors can only occur on $\mathcal{U} \cap \mathcal{E}_1$, a subset of $\mathcal{U}$. By~\cite[Proposition 3.6]{beemer2023network} and~\cite[Proposition~3.9]{beemer2023network}, we have that \[\Omega^i[\mathcal{N},\alp,\mathcal{F},S \longrightarrow T,\mathcal{U},t]) \geq  \Omega_1^i \blacktriangleright \Omega^i[\mathcal{N},\alp,\mathcal{F},\mathcal{E}_1 \longrightarrow \mathcal{E}_2, \mathcal{U}\cap \mathcal{E}_1 ,t] \blacktriangleright \Omega^i_2\] and \[C_1(\Omega^i[\mathcal{N},\alp,\mathcal{F},S \longrightarrow T,\mathcal{U},t])) \leq C_1( \Omega^i[\mathcal{N},\alp,\mathcal{F},\mathcal{E}_1 \longrightarrow \mathcal{E}_2, \mathcal{U}\cap \mathcal{E}_1 ,t]).\] Since we chose $\mathcal{F}$ arbitrarily, this applies for the maximum $\mathcal{F}$, we see that \[C_1(\Omega^i[\mathcal{N},\alp,\mathcal{F},S \longrightarrow T,\mathcal{U},t])) \leq \max_{\mathcal{F}} C_1( \Omega^i[\mathcal{N},\alp,\mathcal{F},\mathcal{E}_1 \longrightarrow \mathcal{E}_2, \mathcal{U}\cap \mathcal{E}_1 ,t])\] as desired.
\end{proof}

By taking $\mE_1 = \mE_2$, we recover a multishot version of~\cite[Corollary 66]{ravagnani2018} as a corollary of Theorem \ref{dcsb}.

\begin{corollary}[Multishot Cut-Set Bound] \label{mulitcut} Let $\mN$ be a network, let $\alp$ be an alphabet, let~$T \in \bold{T}$ and let~$\mE' \subseteq \mE$ be an edge cut. Let $\mU \subseteq \mE$ be a set of edges and let $\mF$ be a network code. We have that
\[\displaystyle C_1(\Omega^i[\mN, \alp, \mF, S \longrightarrow T, \mU, t]) \leq \max_{\mF} C_1(\Omega^i[\mN, \alp, \mF, \mE' \longrightarrow \mE', \mU \cap \mE', t]),\] where the maximum is taken over all network codes $\mF$ for $(\mN, \alp).$
\end{corollary}

Notice that the adversary $\ad_{\mN}$ is restricted to the same proper subset of edges $\mU$ for each uses of the network, but there is no restriction on changing the edges attacked inside of $\mU.$ The case where $\mU$ changes each transmission round remains open. Let $\mN$ be a network and let $\mE_1,\mE_2 \subseteq \mE$ be edge cuts such that $\mE_1$ precedes $\mE_2$. For some $e \in \mE_2$ and~$e' \in \mE_1$, we can say that $e'$ is an \textbf{immediate predecessor} of~$e$ in~$\mE_1$ if~$e' \preccurlyeq e$ and there does not exist $e'' \in \mE_1$ with $e' \preccurlyeq e'' \preccurlyeq e$ and $e' \neq e''$. As in~\cite{beemer2023network}, the vertices of $V_1$ are in bijection with the edges of $\mE_1$ and the vertices of $V_2$ with the edges of $\mE_2$. We say that a vertex $V \in V_1$ is connected to vertex~$V' \in V_2$ if and only if the edge of~$\mE_1$ corresponding to $V$ is an immediate predecessor of the edge of $\mE_2$ corresponding to~$V'$. Let~$\mE_S'$ be the set of edges directly connected with the source of $\mN'$, which can be identified with the edges of $\mE_1$ (consistent with how we identified these with the vertices in $V_1$). The next result is a consequence of Theorem \ref{dcsb} that extends~\cite[Corollary 8.7]{beemer2023network} to the multishot setting and allows for an easier application of the theorem. The proof is inspired by the one of~\cite[Corollary~8.7]{beemer2023network}.

\begin{corollary}\label{corD} Let $\mN$ be a network, $\alp$ a network alphabet, $\mU \subseteq \mE$ a set of edges and~$t \geq 0$.
Let $\mE_1$ and $\mE_2$ be edge cuts between $S$ and $T$ such that $\mE_1$ precedes $\mE_2$ and $\mE_S'$ be as defined above. Let $\mN'$ be a simple $3$-level network. Let $\ad_{\mN}$ be the adversary able to corrupt up to $t$ edges of $\mN$ from $\mU$ and let $\ad_{\mN'}$ be the adversary able to corrupt $t$ edges of $\mN'$ from $\mU\cap \mE'_S$. Then,
\begin{equation}
\C_i(\mN, \alp, \ad_{\mN}) \le \C_i(\mN', \alp, \ad_{\mN'}).
\end{equation}
    
\end{corollary}
\begin{proof}
We wish to show that
\begin{equation*}
    \C_i(\Omega[\mN,\mA,\mF,\mE_1 \longrightarrow \mE_2,\mU \cap \mE_1,t]) \le \C_i(\mN',\mA, \ad_{\mathcal{N}'})
\end{equation*} for every network code $\mF$ for the pair~$(\mN,\mA)$, where $\mF$ is the same for all $i$ transmission rounds. This in turn establishes the corollary due to Proposition~\ref{dcsb}. 
Let $\mF$ be a network code and  let~$\Omega^i:=\Omega^i[\mN,\mA,\mF,\mE_1 \longrightarrow \mE_2,\mU \cap \mE_1,0]$.
By the Remark~\ref{rmk:imm}, we know that the channel $\Omega^i$
expresses the value of each edge $ e \in \mE_2$ as a function of the values of its immediate predecessors in $\mE_1$. 
By the construction of $\mN'$, we can find a network code $\mF'$ (that depends on $\mF$) for $(\mN',\mA)$ such that
\begin{equation} \label{cc1}
\Omega^i=\Omega^i[\mN',\mA,\mF',\mE'_S \longrightarrow \textup{in}(T),\mU \cap \mE'_S,0],
\end{equation}
where the edges of $\mE_1$ and $\mE_2$ are identified with those of $\mE'_S$ and $\textup{in}(T)$ in $\mN'$. We observe that the channel $\Omega^i[\mN,\mA,\mF,\mE_1 \longrightarrow \mE_2,\mU \cap \mE_1,t]$ can be written as the concatenation
\begin{equation} \label{cc2}
\begin{aligned}
\Omega^i[\mN,\mA,\mF,\mE_1 \longrightarrow \mE_2,\mU \cap \mE_1,t] &= (\Omega[\mN,\mA,\mF,\mE_1 \longrightarrow \mE_1,\mU \cap \mE_1,t] \blacktriangleright \Omega)^i \\
&= \Omega^i[\mN,\mA,\mF,\mE_1 \longrightarrow \mE_1,\mU \cap \mE_1,t] \blacktriangleright \Omega^i  
\end{aligned}
\end{equation}
where the first channel in the concatenation simply describes the action of the adversary on the edges of $\mU \cap \mE_1$ over $i$ transmission rounds and the last line comes from~\cite[Proposition~18]{ravagnani2018}.
By \eqref{cc1} and \eqref{cc2} and by identifying between $\mE_1$ and $\mE_S'$,
we can 
\begin{equation}\label{lll}
\begin{aligned}
    \Omega^i[\mN',\mA,\mF',\mE'_S \longrightarrow \textup{in}(T),&\mU \cap \mE'_S,t]\\ &= \Omega^i[\mN',\mA,\mF',\mE'_S \longrightarrow \mE'_S,\mU \cap \mE'_S,t]   \\ 
    &  \qquad \blacktriangleright
\Omega^i[\mN',\mA,\mF',\mE'_S \longrightarrow \textup{in}(T),\mU \cap \mE'_S,0]  \\
&= \Omega^i[\mN,\mA,\mF,\mE_1 \longrightarrow \mE_1,\mU \cap \mE_1,t] \blacktriangleright \Omega^i \\
&=\Omega^i[\mN,\mA,\mF,\mE_1 \longrightarrow \mE_2, \mU \cap \mE_1,t]. 
\end{aligned}
\end{equation}
By definition, we have
$C_i(\mN',\mA,\mU \cap \mE'_S,t) \ge 
C_1(\Omega^i[\mN',\mA,\mF',\mE'_S \longrightarrow \textup{in}(T),\mU \cap \mE'_S,t])$,
and combining this with~\eqref{lll},
leads to
\begin{align*}
C_i(\mN', \alp, \ad_{\mN'}) &\ge C_1(\Omega^i[\mN',\mA,\mF',\mE'_S \longrightarrow \textup{in}(T),\mU \cap \mE'_S,t]) \\ &= 
\C_i(\Omega[\mN,\mA,\mF,\mE_1 \longrightarrow \mE_2,\mU \cap \mE_1,t]).    
\end{align*}
Since  we assumed that $\mF$ was an arbitrary for $(\mN,\mA)$ for $i$ transmission rounds, the result follows.
\end{proof}

\subsection{Multishot Capacity of the Butterfly Network}

We now will apply Corollary \ref{corD} to the Butterfly Network $\mathcal{B}$ in Figure \ref{fig:butt}. The Butterfly Network $\mathcal{B}$ can be reduced to a simple $2$-level network that in turn is the Diamond Network~$\mathcal{D}$~\cite[Section~VIII]{beemer2023network}. The one-shot capacity of $\mathcal{B}$ is $C_1(\mathcal{B}, \ad_{\mathcal{B}}) = \log_{|\alp|}(|\alp|-1)$ according to~\cite[Theorem 8.9]{beemer2023network}. We will show that in Scenario \ref{scenario1},  
\[C_i(\mathcal{B}, \alp, \ad_{\mathcal{B}}) = \frac{\log_{|\alp|}(|\alp|^i -1)}{i},\] which demonstrates a gain in capacity over multiple uses of $\mathcal{B}$ in this scenario. In contrast, we will show that in Scenario \ref{scenario2},
\[C_i(\mathcal{B}, \alp, \ad_{\mathcal{B}}) = \log_{|\alp|}(|\alp| -1)=C_1(\mathcal{B}, \alp, \ad_{\mathcal{B}}).\] 

\paragraph{Scenario~\ref{scenario1} for $\mathcal{B}$.} In this scenario, the adversary is restricted to corrupting an edge of the set $\mU = \{e_1,e_2,e_3,e_4,e_5,e_6,e_9\}$ and cannot change the edge attacked each transmission round. We start with the following proposition that constructs an unambiguous code $C$ with $|C| =  |\alp|^i - 1$ and computes the multishot capacity $\mathcal{B}$. We use the same network code $\mathcal{F}$ as in \cite[Theorem 8.9]{beemer2023network} and extend the strategy in the proof.

\begin{proposition}\label{lbn} Let $i \in \N$ and let $\alp$ be an alphabet. Let $\ad_{\mathcal{B}}$ be an adversary able to corrupt up to an edge of $\mU$ and let $C \subseteq \alp^{4i}$ be an unambiguous code for $(\mathcal{B}, \alp,\mathcal{F})$. Then, the $i$-shot capacity of $\mathcal{B}$ in Scenario \ref{scenario1} is \[C_i(\mathcal{B}, \alp, \ad_{\mathcal{B}}) = \frac{\log_{|\alp|}(|\alp|^i - 1)}{i}.\]
\end{proposition}
\begin{proof} We reserve $\star \in \mathcal{A}$ and want to show that the code \[C = \{(c_1,\ldots,c_i,c_1,\ldots,c_i,c_1,\ldots,c_i,c_1,\ldots,c_i) \in \mA^{4i}\,:\, (c_1,\ldots,c_i) \neq (\star,\ldots,\star)\} \subseteq \alp^{4i}\] is unambiguous for $\mathcal{B}$ for a certain network code $\mF$. We have $c = (a\mid a\mid a \mid a)$, with $a \in \mA^{i}\setminus \{(\star,\ldots,\star)\}$, for any $c \in C$. Let $\mF$ where $V_1$, $V_2$, $V_3$ and $V_4$ proceed as follows: ff the symbols on $V_1$ and $V_2$'s incoming edges are equal, they forwards
that symbol; otherwise they output $(\star, \ldots, \star)$. If one of the two received symbols received by vertex $V_3$ is different from $(\star, \ldots, \star)$,
then $V_3$ forwards that symbol. If both received symbols are
different from $(\star, \ldots, \star)$, then $V_3$ outputs $(\star, \ldots, \star)$ over $e_9$. The vertex $V_4$ just
forwards the symbol received. Suppose the adversary corrupts $e_1$ or $e_2$ and changes the symbol. One can easily show that \[\Omega^i[\mathcal{B}, \alp,\mathcal{F}, S \longrightarrow T, \mU, 1]((a \mid a \mid a \mid a)) = \{((\star,\ldots, \star)\mid a)\,:\,a \in \mA^{i}\setminus \{(\star,\ldots,\star)\}\}\]
for any $a \in \mA^i\setminus \{(\star,\ldots,\star)\}$ and terminal $T_{1}$ trusts the edge $e_{10}$ and terminal $T_2$ trusts the edge $e_{8}$. Similarly, if the adversary corrupts $e_3$ or $e_4$, we have that \[\Omega^i[\mathcal{B}, \alp,\mathcal{F},S \longrightarrow T, \mU, 1]((a\mid a \mid a \mid a)) = \{(a \mid (\star,\ldots, \star))\,:\,a \in \mA^{i}\setminus \{(\star,\ldots,\star)\}\}\] for any $a \in \alp^{i}\setminus\{(\star,\ldots,\star)\}$ and terminal $T_1$ trusts $e_{5}$ and Terminal $T_2$ trusts $e_{11}$. It remains to show that $C$ is unambiguous. Let~$c,c' \in C$. It follows that \[\Omega^i[\mathcal{B},\alp,\mathcal{F},S \longrightarrow T, \mU, 1](c)\cap\Omega^i[\mathcal{B}, \alp,\mathcal{F}, S \longrightarrow T, \mU, 1](c') \neq \emptyset\] for $T \in \bold{T}$ if and only if $c = c'$, Therefore,  implying that $C$ is unambiguous. Lastly, one can easily see that $|C| \geq |\mA|^i -1$. Thus, $C_i(\mathcal{B}, \alp, \ad_{\mathcal{B}}) \geq \frac{\log_{|\alp|}(|\alp|^i - 1)}{i}$.\\

The upper bound follows from Theorem \ref{dcsb} and the observation in~\cite{beemer2023network} that the reduction of~$\mathcal{B}$ to a simple $2$-level network is exactly the Diamond Network $\mathcal{D}$. Therefore, we have that $C_i(\mathcal{B},\alp, \ad_{\mathcal{B}}) \leq C_i(\mathcal{D}, \alp, \ad_{\mathcal{D}})$. By Proposition \ref{diam1}, we therefore have \[C_i(\mathcal{B}, \alp, \ad_{\mathcal{B}}) \leq \frac{\log_{|\alp|}(|\alp|^i - 1)}{i},\]
as desired.
\end{proof}

We observe here that using the Butterfly Network $\mathcal{B}$ multiple times for communication also provides a gain in capacity. This consequence is suggested by the fact that the network's multishot capacity is upper bounded by that of the Diamond Network, which we showed that there is gain in using $\mD$ multiple times for communication in Scenario~\ref{scenario1}.

\paragraph{Scenario~\ref{scenario2} for $\mathcal{B}$.} In this scenario, the adversary can attack an edge from the set~$\mU = \{e_1,e_2,e_3,e_4,e_5,e_6,e_9\}$ and can change the edge attacked each transmission round. We have the following proposition.

\begin{proposition}\label{cbn} Let $\alp$ be an alphabet and let $\ad_{\mathcal{B}}$ be an adversary able to attack up to one edge of $\mU$. Then, the $i$-shot capacity of the Butterfly Network $\mathcal{B}$ in Scenario~\ref{scenario2} is 
    \[C_i(\mathcal{B}, \alp, \ad_{\mathcal{B}}) = \log_{|\alp|}(|\alp| - 1).\]
\end{proposition}
\begin{proof}
    The proof follows from the fact that $C_i(\mathcal{B}, \alp, \ad_{\mathcal{B}}) \leq C_i(\mD, \alp, \ad_{\mD})$  and we showed that~$C_i(\mD, \alp, \ad_{\mD}) = \log_{|\alp|}(|\alp| - 1)$. Therefore,  $C_i(\mathcal{B}, \alp, \ad_{\mathcal{B}}) \leq \log_{|\alp|}(|\alp|-1).$ The lower bound comes from applying the strategy in~\cite[Theorem 8.9]{beemer2023network} independently, which means applying the strategy $i$ times over $i$ transmission rounds.
\end{proof}

The previous result tells us that there is no gain in using $\mathcal{B}$ multiple times for communication in Scenario \ref{scenario2}. Therefore, \[C_i(\mathcal{B}, \alp, \ad_{\mathcal{B}}) = C_1(\mathcal{B}, \alp, \ad_{\mathcal{B}}) = \log_{|\alp|}(|\alp| - 1).\] We note that the results of the Butterfly Network mimic those of the Diamond Network, which is expected since $C_i(\mathcal{B}, \alp, \ad_{\mathcal{B}}) \leq C_i(\mD, \alp, \ad_{\mD})$.

\subsection{Multishot Capacity of General 2-Level Networks}

We now discuss capacity results on general networks $\mN$ that do not meet Theorem \ref{cutset}. There are still two families of networks $\mathfrak{A}_t, \mathfrak{B}_s$ whose multishot capacity was not computed in this paper. Let $a = \min\{|\mE'\setminus \mU| + \max \{|\mE' \cap \mU| - 2t\}\}$. We have the following general result.

\begin{lemma} Let $\mN$ be a simple $2$-level network, $\alp$ be an alphabet and suppose $\mN$ does not meet the Network Singleton Bound, that is, $C_1(\mN, \alp, \ad_{\mN}) < a.$ Let $B$ be the set of reserved vectors in $\alp^{\textup{deg}^{+}(S)}$ such that $|B| = b$. Let $C$ be an unambiguous code for $\mN$, with $C \subseteq \alp^{\textup{deg}^{+}(S)}\setminus B$ and $\mF$ be a network code such that $(C, \mF)$ is a capacity achieving pair for $\mN$, that is, $C_1(\mN, \alp, \ad_{\mN}) = \log_{|\alp|}(|\alp|^a - b).$ Then, the $i$ shot capacity of $\mN$ in Scenario~\ref{scenario1} satisfies 

\[C_i(\mN, \alp, \ad_{\mN}) \geq \log_{|\alp|}(|\alp|^{ai} - b).\]
    
\end{lemma}

\begin{proof} Assume that $\ad_{\mN}$ is an adversary in Scenario~\ref{scenario1} able to corrupt up to $t$ edges of~$\mN$ and cannot change the edges attacked each transmission round. Let \[B^i =\{\underbrace{(k\mid \dots\mid k)}_{i\textup{-times}}\,:\, k \in B\}.\] Notice that $|B^i| = b$. Let $C$ be an unambiguous code for $i$ uses of the network $\mN$ with~$\smash{C \subseteq \alp^{i\textup{deg}^{+}(S)}\setminus B^i}$ and $\mF$ be the same network code in each transmission round. One can check that $|C| \geq |\alp|^{ai} - b$. Therefore,  
\begin{equation*}
    C_i(\mN, \alp, \ad_{\mN}) \geq \frac{\log_{|\alp|}(|\alp|^{ai} - b)}{i}. \qedhere
\end{equation*}
\end{proof}

The previous result tells us that if a strategy exists for the one-shot capacity of families of networks that do not meet the Network Singleton Bound, then in Scenario~\ref{scenario1}, there is gain in using these networks multiple times for communication. 

\begin{corollary}\label{corAB} Let $\alp$ be an alphabet and $\ad_{\mathfrak{A}_t}$ and $\ad_{\mathfrak{B}_t}$ be adversaries in Scenario~\ref{scenario1} for $\mathfrak{A}_t$ and $\mathfrak{B}_t$ respectively. Then, \[C_i(\mathfrak{A}_t, \alp,  \ad_{\mathfrak{A}_t}) \geq \frac{\log_{|\alp|}(|\alp|^{ti} - b)}{i} \hspace{2ex} \hbox{and} \hspace{2ex} C_i(\mathfrak{B_s}, \alp, \ad_{\mathfrak{B}_s}) \geq \frac{\log_{|\alp|}(|\alp|^{si} -b)}{i}.\]    
\end{corollary}

The previous corollary shows that there is a gain in using $\mathfrak{A}_t$ and $\mathfrak{B}_t$ multiple times for communication in Scenario~\ref{scenario1} when a capacity-achieving strategy exists.

\section{Conclusion} \label{S:conclusion}

In this paper, we investigated the multishot capacity of networks with restricted adversaries focusing on the Diamond Network, the Mirrored Diamond Network, and known families of networks as elementary building blocks of a general theory. We extended the Double Cut-Set bound from~\cite{beemer2023network} to the multishot context, allowing for the multishot capacity to be computed of $3$-level networks that can be simplified to $2$-level networks with established one-shot and multishot capacities. 

The results in this paper show that for the Diamond Network $\mD,$ the Butterfly Network~$\mathcal{B}$ and Family $\mathfrak{E}_t$, there is a gain in capacity over multiple uses of the network in Scenario~\ref{scenario1}. The setting in which the adversary is more restricted than in Scenario~\ref{scenario2} and cannot change the edges attacked each transmission round. In Scenario~\ref{scenario2}, where the adversary is more free to change the edges attacked, we have that the one-shot capacity is the same as the $i$-shot capacity (no gain), that is, $C_1(\mN, \alp, \ad_{\mN}) = C_i(\mN, \alp, \ad_{\mN})$ for networks $\mD, \mathfrak{E}_t$ and $\mathcal{B}$. We summarize the best known results on the multishot capacity in Table~\ref{tab:my_label2}.

\begin{table}[h!]
    \renewcommand{\arraystretch}{1.6}
    \centering\resizebox{\columnwidth}{!}{
    \begin{tabular}{|c|c|c|c|}
    \hline 
        \hbox{Network} & $C_i$ \hbox{ in ~\ref{scenario1}} & $C_i$ \hbox{in~\ref{scenario2}} & \hbox{Result}\\
        \hline 
        Diamond Network $\mD$ & $\frac{\log_{|\alp|}(|\alp|^{i}-1)}{i}$  & $\log_{|\alp|}(|\alp|-1)$ & Propositions \ref{diam1}, \ref{diam2}    \\
        \hline
        $\mathfrak{A}_t = ([t,2t],[t,t]), t\geq 2$ & $ \geq \frac{\log_{|\alp|}(|\alp|^{ti}-b)}{i}$ & \hbox{Unknown} & Corollary ~\ref{corAB}\\
        \hline 
        $\mathfrak{B}_s = ([1,s+1],[1,s]), s \geq 1$ & $\geq \frac{\log_{|\alp|}(|\alp|^{si}-b)}{i}$ & \hbox{Unknown} & Corollary ~\ref{corAB}\\
        \hline 
       $\mathfrak{C}_t = ([t,t+1],[t,t]), t \geq 2$ & $ 1$ & $1$ & Proposition \ref{ufc}\\
       \hline
       $\mathfrak{D}_t = ([2t,2t],[1,1]), t \geq 1$ & $ 1$ & $1$ & Proposition~\ref{capacityD}\\
       \hline 
       $\mathfrak{E}_t = ([t,t+1],[1,1]), t \geq 1$ & $ \frac{\log_{|\alp|}(|\alp|^i-b)}{i}$ & $ \log_{|\alp|}(|\alp|-b)$ & Propositions \ref{boundE}, \ref{boundE2}\\
       \hline 
       \hbox{Butterfly Network} $\mathcal{B}$ & $\frac{\log_{|\alp|}(|\alp|^i-1)}{i}$ & $\log_{|\alp|}(|\alp|-1)$ & Propositions ~\ref{lbn}, \ref{cbn}  \\
       \hline

    \end{tabular}}
    \caption{Multishot Capacity of Adversarial Networks.}
    \label{tab:my_label2}
\end{table}

For networks $\mathfrak{C}_t, \mathfrak{D}_t$ and $\mS$, we showed that in both scenarios, there is no gain in using these networks multiple times for communication, regardless of the adversarial model. We note that computing the multishot capacities of Family $\mathfrak{A}_t$ and $\mathfrak{B}_s$ relies on new combinatorial techniques, as the one-shot capacities of these networks has yet to be established. Possible research directions in this area include: A scenario where the subset of edges the adversary can attack also changes over multiple uses; calculating the multishot capacity of networks in the general $n$-level network case; calculating the multishot capacity of networks with more than one source.

\appendix

\section{Some proofs}\label{sec:app}

 In this section, we provide some background for the convenience of the reader. We start with a definition and proposition on finer channels.

\begin{definition}\label{coarser} Let $\Omega_1$,$\Omega_2 : \mathcal{X} \dashrightarrow \mathcal{Y}$ be channels. We say that $\Omega_1$ is finer than $\Omega_2$ (or that $\Omega_2$ is coarser than $\Omega_1$) if $\Omega_1(x) \subseteq \Omega_2(x)$ for all $x \in \mathcal{X}$. In this case, we write~$\Omega_1 \leq \Omega_2$.
\end{definition}

\begin{proposition} \label{finer}\cite[Proposition 3.6]{beemer2023network} If $\Omega_1 , \Omega_2 : \mathcal{X} \dashrightarrow \mathcal{Y}$ are channels with $\Omega_1 \leq \Omega_2$, then $C_1(\Omega_1) \geq C_1(\Omega_2)$.
\end{proposition}

The \textit{concatenation} of channels was introduced in~\cite[Definitions~7 and~14]{ravagnani2018} respectively. Let $\Omega_1:\mathcal{X}_1 \dashrightarrow \mathcal{Y}_1$ and $\Omega_2:\mathcal{X}_2 \dashrightarrow \mathcal{Y}_2$
be channels, and assume that $\mathcal{Y}_1 \subseteq \mathcal{X}_2$. The \textbf{product} of $\Omega_1$ and $\Omega_2$
is the channel $\Omega_1 \times \Omega_2 : \mathcal{X}_1 \times \mathcal{X}_2 \dashrightarrow  \mathcal{Y}_1 \times \mathcal{Y}_2$
defined by   
$$(\Omega_1 \times \Omega_2)(x_1,x_2):= \Omega_1(x_1) \times 
\Omega_2(x_2),$$ $\mbox{ for all 
$(x_1,x_2) \in \mathcal{X}_1 \times \mathcal{X}_2$}$. The \textbf{concatenation} of the channels
$\Omega_1$ and $\Omega_2$ is the channel~$\Omega_1 \blacktriangleright \Omega_2 : \mathcal{X}_1 \dashrightarrow \mathcal{Y}_2$ defined by
$$(\Omega_1 \blacktriangleright \Omega_2)(x):= \bigcup_{y \in \Omega_1(x)} \Omega_2(y).$$

The following result provides insight into the capacity of the concatenation of channels. 
\begin{proposition}[\hspace{-0.1pt}{\cite[Proposition~3.9]{beemer2023network}}]
\label{doi}
Let $\Omega_1:\mathcal{X}_1 \dashrightarrow \mathcal{Y}_1$ and~$\Omega_2:\mathcal{X}_2 \dashrightarrow \mathcal{Y}_2$,
 with~$\mathcal{Y}_1 \subseteq \mathcal{X}_2$, be channels. Then,~$C_1(\Omega_1 \blacktriangleright \Omega_2) \le \min\{C_1(\Omega_1), \, 
C_1(\Omega_2)\}$. 
\end{proposition}

\begin{proposition}[\hspace{-0.1pt}{\cite[Proposition~12]{ravagnani2018}}]
\label{lwpord}
For a channel $\Omega: \mathcal{X} \dashrightarrow \mathcal{Y}$ and any $i\geq 1$, we have that
$$C_1(\Omega^i) \ge i \cdot C_1(\Omega).$$ 
\end{proposition}

\begin{corollary}[\hspace{-0.1pt}{\cite[Corollary 6.2  (Generalized Network Singleton Bound)]{beemer2023network}}] \label{corG} Let $\alp$ be an alphabet, $\mN$ a simple $2$-level network and $\mU_S$ be the set of edges connected to the source. Let an adversary $\ad$ be able to corrupt up to $t$ edges of~$\mN$. Then, the one-shot capacity of~$\mN$ satisfies 
    \begin{equation*}\mathrm{C}_1\left(\mathcal{N}, \ad \right) \leq \min _{P_1 \sqcup P_2=\{1, \ldots, n\}}\left(\sum_{i \in P_1} b_i+\max \left\{0, \sum_{i \in P_2} a_i-2 t\right\}\right)\end{equation*}
with $P_1$ and $P_2$ are two partitions of the set $\{0, \ldots, n\}$ and the minimum is taken over all possible $2$-partitions.
\end{corollary}

\rmv{
\begin{proposition}\label{proofD} Let $\alp$ be an alphabet and let $\ad_{\mathfrak{D}_t}$ be an adversary that can corrupt up to $t$ edges on the first level of $\mathfrak{D}_t.$ Then the i-shot capacity of $\mathfrak{D}_t$ is \[C_i(\mathfrak{D}_t, \ad_{\mathfrak{D}_t}) = 1.\] 
\end{proposition}
\begin{proof} Consider an adversary $\ad_{\mathfrak{D}_t}$ who is restricted to corrupting $t$ edges of the network $\mathfrak{D}_t$ on the first level. Let $\mathcal{F}$ be a network code for $\mathfrak{D}_t$ with $\mathcal{F}: \mA^{4t} \to \mA^{2}$. Assume we use the network $i$ times with network codes $\mathcal{F}_1, \ldots, \mathcal{F}_i$. Then $i$ uses of $\mathfrak{D}_t$ is represented by $\Omega^i_{\mathfrak{D}_t}$. For $x = (x_1, \ldots, x_{4it}) \in \mA^{4t} \times \dotsm \times \mA^{4t},$  \[\Omega^i_{\mathfrak{D}_t}(x) := \h_{\mathfrak{D}_t}^i(\mathcal{F}^1_V(x_1,\ldots,x_{4t}), \ldots, \mathcal{F}^i(x_{4(i-1)t+1}, \ldots, x_{4it}).\] We will show that $C_1(\Omega^i_{\mathfrak{D}_t}) \leq |\mA|^i$. Assume towards a contradiction that there exists an unambiguous code $C$ such that $|C| = |\mA|^i + 1$. Then this implies that we have a unambiguous code \[C' := \{\mF^1(x_1,\ldots,x_{4t}), \ldots, \mF^i(x_{4(i-1)t+1},\ldots,x_{4it}))\} \subseteq  \alp^{4t} \times \dotsm \times \alp^{4t}\] which is unambiguous for $\h_{\mathfrak{D}_t}^i$ of cardinality $|\alp|^i + 1$. Since we have that $|C'| = (|\mA|)^i + 1,$ there must exist $y,y' \in C$, $y \neq y'$ such that they coincide in the first $|\mA|^i$ components, making $C' \in \alp^{4t} \times \dotsm \times \alp^{4t}$ a unambiguous code for $\h_{\mathfrak{D}_t}^i$ which has two different code words that coincide in the first $|\alp|^i$ components. However, by example \ref{conD}, such a code does not exist. Therefore,  $|C| \leq |\alp|^i$. 

The lower bound is an immediate result from applying~\cite[Proposition 12]{ravagnani2018} using the power channel $\Omega^i_{\mathfrak{D}_t}$. In particular, 
\[C_i(\mathfrak{D}_t, \ad_{\mathfrak{D}_t}) = C_i(\Omega_{\mathfrak{D}_t})  = \frac{C_1(\Omega^i_{\mathfrak{D}_t})}{i} \geq \frac{i C_1(\Omega_{\mathfrak{D}_t})}{i} = C_1(\mathfrak{D}_t, \ad_{\mathfrak{D}_t}) = 1.\]
Combining gives the desired result.
\end{proof}}
\end{document}